\documentclass[a4paper,12pt,reqno]{amsart}
\usepackage{amssymb,amsthm, amsmath}
\usepackage{ifthen}
\usepackage{graphicx}
\usepackage{float}
\usepackage{tcolorbox}
\usepackage{diagbox}
\usepackage{graphicx}
\usepackage{subcaption}
\usepackage{listings}  
\usepackage{mathrsfs, bm}
\usepackage{xcolor}            
\usepackage{caption}
\usepackage{mathtools}
\usepackage{algorithm}
\usepackage{algpseudocode}
\usepackage{graphicx}
\usepackage{hyperref}

\definecolor{codegreen}{rgb}{0,0.6,0}
\definecolor{codegray}{rgb}{0.5,0.5,0.5}
\definecolor{codepurple}{rgb}{0.58,0,0.82}
\definecolor{backcolour}{rgb}{0.95,0.95,0.92}

\lstdefinestyle{mystyle}{
    backgroundcolor=\color{backcolour},   
    commentstyle=\color{codegreen},
    keywordstyle=\color{magenta},
    numberstyle=\tiny\color{codegray},
    stringstyle=\color{codepurple},
    basicstyle=\ttfamily\footnotesize,
    breakatwhitespace=false,         
    breaklines=true,                 
    captionpos=b,                    
    keepspaces=true,                 
    numbers=left,                    
    numbersep=6pt,                  
    showspaces=false,                
    showstringspaces=false,
    showtabs=false,                  
    tabsize=4
}
\theoremstyle{plain}
\theoremstyle{plain}

\newtheorem{thm}{Theorem}[section]
\newtheorem*{thm*}{Theorem}

\numberwithin{equation}{section}

\numberwithin{subcase}{case}
\newtheorem{cor}{Corollary}[section]
\newtheorem{lem}{Lemma}[section]

\newtheorem{rem}{Remark}[section]

\newtheorem{defn}{Definition}[section]

\theoremstyle{definition}
\newcounter {own}
\def\theown {\thesection  .\arabic{own}}

{\qed\bigskip}
\DeclareMathOperator*{\sinc}{sinc}

\newcounter{alphabet}

\newcommand{\ds}{\displaystyle}

\newcounter{minutes}
\divide\time by 60
\newcounter{hours}
\multiply\time by 60 \addtocounter{minutes}{-\time}

\begin{document}
\thispagestyle{plain}
\bibliographystyle{amsplain}
\title{Signal Prediction by  Derivative Samples from the Past via Perfect Reconstruction}

\thanks{
File:~AntonySreyaRiya.tex,
          printed: \today,
          \thehours.\ifnum\theminutes<10{0}\fi\theminutes}
\author{SREYA T}
\address{Sreya T, Indian Institute of Technology Dhanbad, Dhanbad-826 004, India.}
\email{sreyasoman8862@gmail.com}
\author{RIYA GHOSH}
\address{Riya Ghosh, Indian Institute of Technology Bombay, Mumbai, Maharashtra-400 076, India.}
\email{riya74012@gmail.com}
\author{A. ANTONY SELVAN$^\dagger$}
\address{A. Antony Selvan, Indian Institute of Technology Dhanbad, Dhanbad-826 004, India.}
\email{antonyaans@gmail.com}

\keywords{Averaged modulus of smoothness, Block Laurent operators, B-splines, Derivative sampling, Shift-invariant spaces, Zak transform.\\
$^\dagger$ {\tt Corresponding author}
}

\begin{abstract}
This paper investigates signal prediction through the perfect reconstruction of signals from shift-invariant spaces using nonuniform samples of both the signal and its derivatives. The key advantage of derivative sampling is its ability to reduce the sampling rate.  We derive a sampling formula based on periodic nonuniform sampling (PNS) sets with derivatives in a shift-invariant space. We establish the necessary and sufficient conditions for such a set to form a complete interpolating sequence (CIS) of order $r-1$. This framework is then used to develop an efficient approximation scheme in a shift-invariant space generated by a compactly supported function. Building on this, we propose a prediction algorithm that reconstructs a signal from a finite number of past derivative samples using the derived perfect reconstruction formula. Finally, we validate our theoretical results through practical examples involving cubic splines and the Daubechies scaling function of order 3. 
\end{abstract}
\maketitle
\pagestyle{myheadings}
\markboth{Sreya T, Riya Ghosh and A. Antony Selvan}{Signal Prediction by  Derivative Samples from the Past via Perfect Reconstruction}
\section{Introduction}
The Whittaker-Kotel'nikov-Shannon (WKS) sampling theorem states that any bandlimited signal $f$ of bandwidth $\pi W$ can be  reconstructed from its samples 
$\{f(n/W):n\in\mathbb{Z}\}$ by the sampling formula
\begin{align}\label{shannon}
f(t)=\ds\sum_{n\in\mathbb{Z}}f\left(\frac{n}{W}\right)\sinc{(Wt-n)},~\sinc{t}=\dfrac{\sin \pi t}{\pi t}.
\end{align}
However, practical considerations reveal several limitations in the direct application of the WKS sampling theorem: 
\begin{itemize}
\item [($i$)]  The sinc kernel has infinite support and decays slowly, which makes the space of bandlimited signals often unsuitable for numerical implementations.
\item [($ii$)]  In reality, only a finite number of samples $N$ can be taken. 
Truncating the infinite series \eqref{shannon} to a finite sum introduces errors that decay slowly as $N\to\infty$ again due to the slow decay of the $\mathrm{sinc}$ kernel.
\item [($iii$)] Exact bandlimitation is a strong and often impractical assumption. In reality, no signal can be both strictly timelimited and bandlimited due to the uncertainty principle. Moreover, the WKS sampling series may fail to converge for all continuous signals.
\item [($iv$)]  The reconstruction formula requires samples from the entire time axis, both past and future, relative to a given point in time 
$t=t_0.$ However, in practical applications, only past values 
$f(t)$ for $t<t_0$ are available. This gives rise to the fundamental problem of prediction: Is it possible to reconstruct a bandlimited signal using only past samples?
\end{itemize}
These challenges motivate the exploration of alternative reconstruction techniques, especially those suited for efficient numerical implementation. 

A shift-invariant space  $V(\varphi)$ can be used as a universal model for sampling and interpolation problems because it can include a wide range of functions, whether bandlimited or not, by appropriately choosing a stable generator $\varphi.$ It is defined as 
$$V(\varphi):=\left\{f\in L^2(\mathbb{R}): f(\cdot)=\sum\limits_{k\in\mathbb{Z}}d_k\varphi(\cdot-k)~\text{for some}~ (d_k)\in \ell^2(\mathbb{Z})\right\}.$$
Recall that $\varphi$ is said to be a stable generator for $V(\varphi)$ if $\{\varphi(\cdot-k) : k \in \mathbb{Z}\}$ is a Riesz basis for $V(\varphi)$, \textit{i.e.},
$\overline{span}\{\varphi(\cdot-k): k\in\mathbb{Z}\}=V(\varphi)$ and there exist constants $A$, $B>0$ such that
\begin{equation}\label{rieszbasis}
A\sum_{k\in\mathbb{Z}}|d_k|^2\leq\big\|\sum_{k\in\mathbb{Z}}d_k\varphi(\cdot-k)\big\|_2^2\leq
B\sum_{k\in\mathbb{Z}}|d_k|^2,
\end{equation}
for all $(d_k)\in\ell^2(\mathbb{Z})$.    
It is well known that $\varphi$ is a stable generator for $V(\varphi)$ if and only if
\begin{equation}\label{eqn2.5}
 0<\|\Phi\|_0 \leq \|\Phi\|_{\infty} <\infty, 
\end{equation}
where $\|\Phi\|_0$ and $\|\Phi\|_\infty$ denote the essential infimum and supremum of the function
$\Phi(w):=\sum_{n\in\mathbb{Z}}|\widehat{\varphi}(w+n)|^2$ in the interval $[0,1],$ respectively. Here, $\widehat{\varphi}$ denotes the Fourier transform of $\varphi$, which is defined as
$$\widehat \varphi(w):= \int_{-\infty}^{\infty} \varphi(x) e^{-2\pi \mathrm{i}wx}~dx.$$
In addition,
if $\varphi$ is a continuously $(r-1)$-times differentiable function such that for some $\epsilon>0,$
$$\varphi^{(s)}(t)=\mathcal{O}(|t|^{-1-\epsilon})~\text{as}~t\to \pm\infty,~0\leq s\leq r-1,$$
then we say that $\varphi$ is a $(r-1)$-regular stable generator for $V(\varphi).$
In this case, every $f\in V(\varphi)$ is $(r-1)$-times differentiable and 
$$f^{(s)}(t)=\sum\limits_{k\in\mathbb{Z}}c_{k}\varphi^{(s)}(t-k),~0\leq s\leq r-1.$$

Shift-invariant spaces, initially introduced in approximation theory and wavelet theory, extend the concept of bandlimited signals. Sampling in shift-invariant spaces offers a practical model for a variety of applications, including the consideration of real-world acquisition and reconstruction devices, the modeling of signals with smoother spectra than bandlimited signals, and numerical implementations. For further details, we refer to \cite{AlGr} and the references therein.

Beyond sampling only the signal values, the concept of sampling both a signal and its derivatives is well established in communication engineering, information theory, and signal processing.  This technique offers the primary advantage of reducing the required sampling rate. Shannon first introduced this idea in his seminal 1949 work \cite{Shannon1949}.
To explain the problem of sampling with derivatives in $V(\varphi)$ more precisely, let us first introduce the following terminology.
\begin{defn}
Let $\varphi$ be a $(r-1)$-regular stable generator for $V(\varphi)$. Then a set $\Gamma=\{x_n: n\in\mathbb{Z}\}$ of real numbers is said to be 
\begin{enumerate}
\item[$(i)$] a set of stable sampling (SS) of order $r-1$ for $V(\varphi)$ if there exist constants $A, B>0$ such that $$A\|f\|_{2}^2\leq\sum\limits_{n\in\mathbb{Z}}\sum\limits_{i=0}^{r-1}|f^{(i)}(x_n)|^2\leq B\|f\|_{2}^2,$$
for all $f\in V(\varphi);$
\item[$(ii)$] a set of interpolation (IS) of order $r-1$ for $V(\varphi)$ if the interpolation problem $$f^{(i)}(x_n)=c_{ni}, n\in\mathbb{Z}, i= 0, 1, \dots, r-1,$$ has a solution $f\in V(\varphi)$ for every square summable sequence $\{c_{ni}: n\in\mathbb{Z}, i= 0, 1, \dots, r-1\}.$
\end{enumerate}
If $\Gamma$ is both  SS and IS of order $r-1$ for $V(\varphi),$  then it is called a complete interpolation set (CIS) of order $r-1$ for $V(\varphi).$
\end{defn}

Fogel \cite{Fogel}, Jagermann and Fogel \cite{JagermanFogel}, Linden and Abramson \cite{LindenAbramson} proposed the first work on sampling with derivative values in the space of bandlimited signals. They proved that $r\mathbb{Z}$ is a CIS of order $r-1$  for $V(\textrm{sinc})$.
The problem of sampling with derivatives in bandlimited signal spaces was initially studied in \cite{Grochenig1, Papoulis, Raza, Nathan1973} and more recently revisited with an emphasis on shift-invariant spaces in \cite{Adcock, Grochenig2, Selvan, Priyanka,Priyanka2025}.
The use of derivative-based sampling has found applications in various domains, including aircraft instrumentation, air traffic management simulations, wireless communication in minimizing energy consumption, and telemetry \cite{Fogel}.  In the mid-1980s, Corey and O'Neil demonstrated this concept through its application to planar near-field antenna measurements \cite{Corey1986}. The concept has also found use in the design of wideband fractional delay filters \cite{Tseng}. More recently, between 2022 and 2024, researchers have extended derivative sampling to spherical near-field antenna measurements, as reported in \cite{Kaslis2022, Kaslis2023a, Kaslis2023b, Kaslis2024, Kaslis2024b, Breinbjerg2024}.

Djokovic and Vaidyanathan \cite{Djokovic1997} explored generalized sampling schemes within the framework of shift-invariant spaces, including methods such as local averaging and derivative sampling.  Building on earlier work by Papoulis \cite{Papoulis} in the bandlimited setting, Unser and Zerubia \cite{UnZer2} addressed the challenge of identifying the conditions under which any function $f\in V(\varphi)$ can be reconstructed from such generalized samples.  Although both \cite{Djokovic1997} and \cite{UnZer2} established necessary and sufficient conditions for stable reconstruction, practical criteria for the existence of interpolating kernels and efficient methods for their computation remained unresolved. These challenges are further explored and addressed in recent developments, particularly in the works \cite{Hogan1} and \cite{Hogan2}.

A particularly useful class of sampling sequences, called periodic nonuniform sampling (PNS) sets, is given by 
$\{x_n+\rho l :~1\leq n\leq L,~l\in\mathbb{Z}\}$, with period $\rho$ and length $L$. The PNS pattern is determined by an offset vector $x=(x_0, x_1, \dots,x_{L-1}).$ Periodic nonuniform sampling, also known as bunched or interlaced sampling, dates back to the works of Kohlenberg \cite{Kohlenberg} and Yen \cite{yen}. 
PNS has found important applications in areas such as channel noise reduction \cite{Aragones} and the analysis of multiband signals \cite{Herley}. Nathan \cite{Nathan1973} investigated PNS methods that utilize both function values and their derivatives for reconstructing bandlimited signals. The authors of \cite{StrohmerTanner} proposed a method for reconstructing a bandlimited signal from periodically nonuniform samples, achieving root-exponential accuracy using only a finite number of samples. A key motivation for using PNS is its ability to achieve perfect reconstruction, meaning exact recovery of the original signal without distortion or information loss. In shift-invariant spaces, this requires both the generator and interpolating kernels to be compactly supported, ensuring a finite reconstruction sum with no truncation error.
The main challenge lies in selecting sampling sets that ensure this compact support. 
This contrasts with the classical bandlimited case, where the ideal sinc kernel has infinite support, and any practical reconstruction inevitably involves truncating an infinite sum, introducing unavoidable approximation errors. The authors in \cite{Djokovic1997} and \cite{UnZer2} established a fundamental link between filter bank theory and sampling in shift-invariant spaces, which is vital for achieving perfect reconstruction. In \cite{Garcia0}, Garcia et al. further explored an approximation scheme based on the perfect reconstruction process. Another important motivation for PNS is its relevance to the local reconstruction problem, which is closely related to perfect reconstruction and has been explored in the literature (see \cite{Sun2010, sun2, Vaidyanathan2002}). The necessary and sufficient conditions for perfect reconstruction, along with its classification, are explained in \cite{Martin, Martin2}.

The problem of perfect reconstruction using derivative samples has not been extensively explored in the literature. The first objective of this paper is to establish a sampling formula based on the PNS set with derivatives in the shift-invariant space $V(\varphi)$ and to demonstrate how this leads to perfect reconstruction. We provide a necessary and sufficient condition for the PNS set to be a CIS of order $r-1$ for $V(\varphi)$. Consequently, we show that every $f\in V(\varphi)$ can be reconstructed from  the following sampling formula:
\begin{align*}
f(t)=\sum_{n=0}^{L-1}\sum_{l\in\mathbb{Z}}\sum\limits_{i=0}^{r-1}f^{(i)}(x_n+ \rho l)\Theta_{ni}(t-\rho l),
\end{align*}
where the interpolating kernels $\Theta_{ni}$ can be computed explicitly. When the generator $\varphi$ is compactly supported, we prove that a suitable choice of the offset vector in the PNS framework enables the construction of compactly supported interpolating kernels that guarantee perfect reconstruction (Theorem \ref{theorem3.2}). 
Additionally, we use our main result to investigate the complete interpolation property of specific sampling sets for shift-invariant spline spaces, which are particularly valuable in practical applications due to the availability of efficient and robust algorithms for their computation.
While previous work \cite{Djokovic1997,BBSV2} addressed perfect reconstruction in spline spaces, it did not consider offset selection, derivative sampling, or other shift-invariant spaces like those generated by Daubechies scaling functions.
In this work, we extend these results to derivative sampling and arbitrary compactly supported generators, and we present a systematic method for selecting the offset vector in both classical and derivative sampling cases. This approach ensures that PNS based sampling schemes operate at minimal sampling rates, avoid truncation errors, and maintain low computational complexity, guaranteeing perfect reconstruction even when derivatives are involved. 

In general, signals do not naturally belong to a specific shift-invariant space. Nevertheless, as Lei et al. emphasized in \cite{Lei}, a variety of methods have been proposed in the literature to construct approximation schemes within such spaces. These methods include cardinal interpolation, quasi-interpolation, projection, and convolution (see \cite{DDR, Jia, JiaLei}). Moreover, in the works of Garcia et al. \cite{Garcia1, Garcia2, Garcia3}, the authors examined the problem of approximating functions in shift-invariant spaces using generalized sampling formulas. In these and related studies, the considered signals were typically assumed to be differentiable to a certain order, and the approximation error was analyzed using the classical modulus of smoothness. However, real-world signals are often not continuous and may feature jump discontinuities, such as those caused by shocks from flying missiles, seismic shocks, or extrasystoles in heartbeats. To address this, the classical modulus of smoothness, which is typically used for continuous functions in $C(\mathbb{R})$, needs to be replaced with a modulus of smoothness that can handle discontinuous signals.  The Bulgarian school, under the guidance of Sendov \cite{Sendov}, introduced the concept of an averaged modulus of smoothness, which provides a means of estimating aliasing errors for functions defined on an interval $[a,b]$ that may not necessarily be differentiable. In \cite{BBSV1}, the authors extended this concept for functions defined on the real axis. Consequently, the authors \cite{BBSV1,BBSV2} investigated the rate of approximation of a signal $f$ by the following sampling operator
\begin{align*}
   S_Wf(t)=\ds\sum_{n\in\mathbb{Z}}f\left(\frac{n}{W}\right)\varphi(n-Wt) 
\end{align*}
based on the averaged modulus of smoothness $\tau_r(f;W^{-1})_p$ instead of  the classical modulus of continuity. More recently, the authors in \cite{Selvan1} studied the convergence behavior of sampling expansions in shift-invariant spaces generated by a Gaussian function. Inspired by the work in \cite{BBSV1, BBSV2} and the practical purpose of derivative sampling, the authors \cite{Priyanka} investigated the rate of approximation of a signal $f$ (not necessarily continuous) by the derivative sampling series based on spline generators.

The second aim of this paper is to demonstrate how the proposed sampling formula based on the periodic nonuniform samples of a function in the space $V(\varphi)$ along with its derivatives can be effectively utilized to develop an effective approximation scheme. In particular, we construct an approximation scheme using a derivative sampling formula based on PNS in a shift-invariant space generated by a compactly supported generator $\varphi$.
To this end, we define an approximation operator $S_W$ for a suitable real-valued function $f$ using the following sampling series:
\begin{align}
    (S_Wf)(t):= \sum\limits_{l\in\mathbb{Z}}\sum\limits_{i=0}^{r-1}\sum\limits_{n=0}^{L-1}\frac{1}{W^i}f^{(i)}\left(\frac{x_n+\rho l}{W}\right)\Theta_{ni}(Wt-\rho l).
\end{align}

In this paper, we investigate the rate of convergence of the sampling series $S_Wf$ to the approximation of a signal $f$, which may not necessarily be continuous (Theorem \ref{approximation operator}).  We achieve this by analyzing the reproducing polynomial property of the sampling series, which plays a critical role in determining the accuracy of the approximation. 

As mentioned earlier, the WKS sampling formula requires samples from both the past and future, but in practice, only a finite number of past samples are available, raising the question of whether a signal can be reconstructed from these limited past samples. The problem of predicting a signal based on its past samples traces its origins to the work of L. A. Wainstein and V. D. Zubakov \cite{Wainstein} and J. L. Brown \cite{Brown}. This area of research was subsequently advanced by Splettst\"o{\ss}er, who made significant contributions to the theoretical framework and practical methodologies of signal prediction. For further insights and developments in this field, particularly those building on Splettst\"o{\ss}er's work, the reader is encouraged to consult the publications by Mugler and Splettst\"o{\ss}er \cite{Mugler1, Mugler2, Mugler3, Splettstoeser1, Splettstoeser2}, which provide detailed discussions and extensions of the early models. Later, Butzer and Stens addressed the prediction problem, adding valuable theoretical insights and practical applications, as discussed in their work \cite{Butzer2, Butzer3, Butzer4}. For a more comprehensive overview of the prediction problem, the extensive review in \cite{Butzer4} provides valuable insights for readers.

The main goal of this paper is to explore how to predict a signal from a finite number of past derivative samples. We demonstrate that our perfect reconstruction formula enables the prediction of the signal from these limited past samples. By appropriately shifting the support of the interpolating kernels in the reconstruction formula, we can effectively predict the signal from the available finite past derivative samples. Based on this idea, we propose a prediction algorithm derived from our results to offer an efficient method for predicting signals from a finite number of past samples. Finally, we validate our theoretical results by applying the method to cubic splines and the Daubechies scaling function of order 3. Simulations demonstrate the effectiveness of the proposed prediction algorithm when applied to these specific generators.

The paper is organized as follows. Section 2 introduces the Zak transform, block Laurent operators, and the average modulus of smoothness—fundamental concepts in the theory of periodic nonuniform sampling in shift-invariant spaces. Section 3 focuses on perfect reconstruction based on periodic nonuniform sampling with derivatives in these spaces. Section 4 explores how the perfect reconstruction formula leads to a prediction algorithm that estimates a signal from its finite number of past derivative samples. Finally, detailed proofs of the main results are presented in the Appendix to ensure readability and maintain the flow of the main exposition.

\section{Preliminaries}
 In this section, we introduce the Zak transform and block Laurent operators, which play a crucial role in characterizing the sampling and interpolation of periodic non-uniform samples in shift-invariant spaces. We also present the averaged modulus of smoothness, which is used to estimate the approximation error for discontinuous signals.
\subsection{\texorpdfstring{\underline{The Zak Transform}}{The Zak Transform}}

For a given parameter $\alpha>0,$ the Zak transform $\mathcal{Z}_\alpha f$ of a function $f$ on $\mathbb{R}^2$ is defined by 
$$\mathcal{Z}_\alpha f(x,y)=\sum\limits_{k\in\mathcal{Z}}f(x-\alpha k)e^{2\pi \mathrm{i}\alpha ky}.$$ 
When $\alpha=1$, we use the notation  $\mathcal{Z}$ instead of $\mathcal{Z}_1.$
A continuous function $f$ belongs to the Wiener space $W(\mathbb{R})$ if $$\|f\|_{W(\mathbb{R})}:=\sum_{n\in\mathbb{Z}}\max\limits_{x\in[0,1]}|f(x+n)|<\infty.$$
If $f$ and $\widehat{f}$ belong to $W(\mathbb{R})$, then the Poisson summation formula \cite{time} 
\begin{align}\label{poissonsum}
\sum_{n\in\mathbb{Z}}f(x+\nu n)=\dfrac{1}{\nu}\sum_{n\in\mathbb{Z}}\widehat{f}\left(\dfrac{n}{\nu}\right)e^{2\pi \mathrm{i} nx/\nu},~\nu>0,
\end{align}
holds for all $x\in\mathbb{R}$ with absolute convergence of both sums. As a consequence, the Zak transform of a function $f$ and its Fourier transform $\widehat{f}$ are related by
\begin{align}\label{zak}
\mathcal{Z}_\alpha f(x,y)=\tfrac{e^{2\pi \mathrm{i}xy}}{\alpha}\mathcal{Z}_{\tfrac{1}{\alpha}}\widehat{f}(y,-x),~\text{for}~x,y\in\mathbb{R}.
\end{align}
If $f\in W(\mathbb{R})$, then $\mathcal{Z}_{\alpha} f$ is continuous on $\mathbb{R}^2$.
We refer to \cite{time} for other properties of the Zak transform.
\subsection{\texorpdfstring{\underline{Block Laurent Operators}}{Block Laurent Operators}}

Let $\mathcal{A}$ be a bounded linear operator on $\ell^2(\mathbb{Z})$ with
the associated matrix given by $\mathcal{A}=[a_{rs}]$. Then $\mathcal{A}$ is said to be a \textit{Laurent operator} if the matrix entries satisfy the condition $a_{r-k,s-k}=a_{rs}$, for every $r$, $s$, $k$
$\in\mathbb{Z}$. For $\varphi\in L^\infty(\mathbb{T})$, let us define $\mathcal{M}:L^2(\mathbb{T})\to
L^2(\mathbb{T})$ by 
\[(\mathcal{M}f)(x):=\varphi(x)f(x), ~\text{for}~a.e.~x\in\mathbb{T},\]
where $\mathbb{T}$ is the circle group identified with the unit interval  $[0,1)$. Then the operator $\mathcal{A}=\mathcal{F}\mathcal{M}\mathcal{F}^{-1}$ is a Laurent operator defined by the matrix $[\widehat{\varphi}(r-s)]$, where $\mathcal{F}:L^2(\mathbb{T})\to \ell^2(\mathbb{Z})$ is the Fourier transform defined by 
$$\mathcal{F}f=\widehat{f},~\widehat{f}(n)=\int_{0}^{1} f(x) e^{-2\pi \mathrm{i}n x}dx,~n\in\mathbb{Z}.$$ 
In this case, we say that $\mathcal{A}$ is a Laurent operator defined by the symbol $\varphi$. If it is invertible, then $\mathcal{A}^{-1}$ is also a Laurent operator with the symbol $\frac{1}{\varphi}$ and the matrix of $\mathcal{A}^{-1}$ is given by
$$\mathcal{A}^{-1}=\left[\widehat{\left(\tfrac{1}{\varphi}\right)}(r-s)\right].$$
Let $\mathcal{L}$ be a bounded linear operator on $\ell^2(\mathbb{Z})^m$ which is represented by a double infinite matrix whose entries are $m\times m$ matrices, \textit{i.e.},
$$L=[A_{rs}],~ A_{rs} ~\text{is an}~ m\times m ~ \text{matrix.}$$
Since $\ell^2(\mathbb{Z})^m$ is a direct sum of $m$ copies of $\ell^2(\mathbb{Z})$, an operator $\mathcal{L}$ on $\ell^2(\mathbb{Z})^m$ may also be represented by an $m\times m$ matrix whose entries are doubly infinite matrices, \textit{i.e.},
\begin{align*}
 \mathcal{L}=\left(
\begin{array}{ccccccc}
 L_{11} & L_{12}  & \cdots & L_{1m} \\
L_{21} & L_{22}  &  \cdots & L_{2m} \\
\vdots & \vdots  & \ddots & \vdots\\
L_{m1} & L_{m2}  & \cdots & L_{mm}  
 \end{array}
\right), ~ L_{rs} ~\text{is a doubly infinite matrix}.
\end{align*}
If $\mathcal{L}$ is a block Laurent operator on $\ell^2(\mathbb{Z})^m$ with the above two matrix representations, then
$$L_{rs}=[A_{i-j}^{rs}]_{i,j\in\mathbb{Z}},$$
where $A_n^{rs}$ is the $(r,s)$-th entry of the matrix $A_n=A_{n,0}$ with respect to the standard orthonormal basis of $\mathbb{C}^m$. 

For an $m\times m$ matrix-valued function $\Phi \in L_{m\times m}^1(\mathbb{T})$, the Fourier coefficient $\widehat{\Phi}(k)$ of $\Phi$ is an $m\times m$ matrix is defined as 
$$ \widehat \Phi (k):= \int\limits_{\mathbb{T}}  e^{-2\pi \mathrm{i} k  x}\Phi(x)~dx, ~ k \in \mathbb{Z},$$
where $(i,j)$-th entry is equal to the Fourier coefficient of the $(i,j)$-th entry of $\Phi$. It is well known that for every block Laurent operator $\mathcal{L}$, there exists a function $\Phi \in L^{\infty}_{m\times m}(\mathbb{T})$ such that $\mathcal{L}=[\widehat{\Phi}(i-j)]$ (see \cite[p. 563-566]{GoGoKa1}). In this case,  $\mathcal{L}$ is called a block Laurent operator defined by $\Phi$. 
\begin{thm}\label{IB} \cite{GoGoKa1}.
Let $\mathcal{L}=[L_{r,s}]=[\widehat{\Phi}(i-j)]$ denote the matrix of a function $\Phi\in L^{\infty}_{m\times m}(\mathbb{T})$. Then $\mathcal{L}$ is an invertible block Laurent operator with the symbol $\Phi\in L^\infty_{m\times m}(\mathbb{T})$ if and only if there exist  two positive constants $A$ and $B$ such that
\begin{equation}\label{Invertibility condition}
 A\leq |\det\Phi(x)|\leq B,\text{ for a.e. }x\in\mathbb{T}.
\end{equation}
\end{thm}

\subsection{\texorpdfstring{\underline{Averaged Modulus of Smoothness}}{Averaged Modulus of Smoothness}}
Let $M(\mathbb{R})$ denote the space of all bounded measurable functions on $\mathbb{R}$, and let $AC_{loc}^{(r)}$ denote the space of all $r$-fold locally absolutely continuous functions on $\mathbb{R}$.
The Sobolev spaces $W^r_p\equiv W^r(L^p(\mathbb{R})),$ $r\in\mathbb{N},$ $1\leq p<\infty,$ are given by
$$W_p^r:=\left\{f\in L^p(\mathbb{R}):  f\in AC^{(r)}_{loc}(\mathbb{R}), f^{(r)}\in L^p(\mathbb{R})\right\}.$$
Let $$(\Delta_h^rf)(t):=\sum\limits_{j=0}^r(-1)^{r-j} \binom{r}{j} f(t+jh)$$ be the classical finite forward difference of order $r\in\mathbb{N}$ of $f$ with increment $h$ at the point $t$. For a positive integer $r$, the 
$r$-th modulus of smoothness of a function $f\in L^p(\mathbb{R})$ is defined as 
$$\omega_r(f;\delta)_p:=\sup\limits_{0\leq h\leq\delta}\|\Delta_h^rf\|_p, ~\delta\geq 0.$$
Here and throughout this paper, $\|\cdot\|_p$ denotes the $L^p$-norm on the real line.
For $f\in M(\mathbb{R})$, the local modulus of smoothness of order $r\in\mathbb{N}$ at the point $x\in\mathbb{R}$
is defined  for $\delta\geq 0$ by
$$\omega_r(f,x;\delta):=\sup\left\{|\Delta_h^rf(t)|:t,t+rh\in\left[x-\dfrac{r\delta}{2},x+\dfrac{r\delta}{2}\right]\right\}.$$
For $f\in M(\mathbb{R})$, $1\leq p<\infty,$ the $L^p$-averaged modulus of smoothness of order $r\in\mathbb{N}$ is defined by
$$\tau_r(f;\delta)_p:=\|\omega_r(f,\bullet;\delta)\|_p,~\delta>0.$$

We now introduce the space $\Lambda^p_{r}$ for $1\leq p<\infty$, which includes Sobolev spaces as well as some discontinuous functions.
For fixed  $ \rho$ and $L$, we consider the sampling set $$X = \left\{x_n + \rho l:~0\leq n\leq L-1,~l\in\mathbb{Z} \right\}.$$
Let $f$ be a  $(r-1)$-times differentiable function on $\frac{1}{W} X$.  Then the discrete $\ell^p_{r}(W)$-norm of $f$ is defined for $1\leq p<\infty$ by
\begin{equation}\label{lpsigma}
\|f\|_{\ell^p_{r}(W)}
:= \left( \sum_{l \in \mathbb{Z}} \sum_{n=0}^{L-1} \sum_{i=0}^{r-1} 
\left|f^{(i)}\left( \frac{x_n + \rho l}{W} \right)\right|^p \frac{1}{W} \right)^{1/p}.
\end{equation}
The space $\Lambda^p_{r}$ is defined by
$$
\Lambda^p_{r} := \left\{ f \in M(\mathbb{R}) : \|f\|_{\ell^p_{r}(W)}
<\infty, ~\text{for every} ~ W \geq 1 \right\}.
$$
This space was introduced in \cite{Priyanka} for the general sampling set $X$. Arguing as in  \cite{BBSV1}, we can prove that $W_p^r \subset\Lambda^p_{r}$.
We now cite a particular case of the basic interpolation theorem for the averaged modulus of smoothness from \cite{Priyanka}.
\begin{thm}\label{interpolation theorem}\cite{Priyanka}.
Let $(L_W)_{W \geq 1}$ be a family of linear operators mapping from $\Lambda^p_{r}$ into $L^p(\mathbb{R})$, $1\leq p<\infty,$ satisfying the following properties:
\begin{itemize}
\item [$(i)$] $\|L_Wf\|_p\leq K_1\|f\|_{\ell^p_{r}(W)}$, for every $f\in \Lambda^p_{r}$, $W\geq 1$,
\item [$(ii)$] $\|L_Wg-g\|_p\leq K_2W^{-r}\|g^{(r)}\|_p$, for every $g\in W_p^r$, $W\geq 1$,
\end{itemize}
for the constants $K_1$, $K_2$ depending only on $r$. Then for each $f\in\Lambda^p_{r}$, there holds the error estimate 
\begin{equation}\label{interpolationthm}
\|L_{W}f-f\|_p\leq \sum\limits_{i=0}^{r-1}c_i\tau_r\left(f^{(i)};\dfrac{1}{W}\right)_p, ~ \text{for every}~ W\geq 1,
\end{equation}
where the constants $c_i,$ $i=0,1,\dots, r-1$  depend only on $K_1$ and $K_2.$ 
\end{thm}

\section{Perfect Reconstruction via Periodic Nonuniform Sampling with Derivatives}
Let $\mathcal{A}_r$ denote the class of continuously $(r-1)$-times differentiable functions such that for some $\epsilon>0$ and $C >0$,
\begin{equation}\label{decayproperty}
   |\varphi^{(s)}(t)|\leq C  (1+|t|)^{-1-\epsilon}\text{ and }
 |\widehat{\varphi^{(s)}}(\omega)|\leq C  (1+|\omega|)^{-1-\epsilon},
 ~0\leq s\leq r-1. 
\end{equation}
If $\varphi\in \mathcal{A}_r$ is a stable generator for $V(\varphi),$ then every $f\in V(\varphi)$ is $(r-1)$-times differentiable and 
\begin{align*}
f^{(p)}(t)=\sum_{k\in\mathbb{Z}}c_{k} \varphi^{(p)}(t-k)=\sum\limits_{k\in \mathbb{Z}}\sum\limits_{q=0}^{\rho-1} c_{\rho k+q}\varphi^{(p)}(t-\rho k-q), ~p=0,1,\dots,r-1.
\end{align*}

Let us consider a PNS set $X$ of period $\rho$ with length $L$, given by
$$X=\{x_n+\rho\ell:~0\leq n\leq L-1,~l\in\mathbb{Z}\}.$$  Associated with this sampling set, we define the following  infinite system 
\begin{align}\label{Eqn2.4}
f^{(p)}(x_n +\rho l)=\sum\limits_{k\in \mathbb{Z}}\sum\limits_{q=0}^{\rho-1} c_{\rho k+q}\varphi^{(p)}(x_n +\rho(l-k)-q),
\end{align}
for $l\in \mathbb{Z},~n = 0, 1, \dots, L-1 ,~p=0,1,\dots,r-1$. To write this system compactly, let us define
\begin{align}
U &=
\begin{bmatrix}
U_{0}^{00} & U_{0}^{01} & \cdots & U_{0}^{0,\rho-1}\\
U_{0}^{10} & U_{0}^{11} & \cdots & U_{0}^{1,\rho-1}\\
\vdots & \vdots & \ddots & \vdots\\
U_{0}^{(r-1),0} & U_{0}^{(r-1),1} & \cdots & U_{0}^{(r-1),(\rho-1)}\\
\hline
\vdots & \vdots & \ddots & \vdots\\
\hline
U_{L-1}^{00} & U_{L-1}^{01} & \cdots & U_{L-1}^{0,\rho-1}\\
U_{L-1}^{10} & U_{L-1}^{11} & \cdots & U_{L-1}^{1,\rho-1}\\
\vdots & \vdots & \ddots & \vdots\\
U_{L-1}^{(r-1),0} & U_{L-1}^{(r-1),1} & \cdots & U_{L-1}^{(r-1),(\rho-1)}
\end{bmatrix}, \nonumber 
 C = 
\begin{bmatrix}
C_0\\
C_1\\
C_2\\
\vdots\\
C_{\rho-1}\\
\end{bmatrix},~
\text{and}~
F =
\begin{bmatrix}
F_0^{(0)}\\
F_0^{(1)}\\
\vdots\\
F_0^{(r-1)}\\
\hline
\vdots\\
\hline
F_{L-1}^{(0)}\\
F_{L-1}^{(1)}\\
\vdots\\
F_{L-1}^{(r-1)}
\end{bmatrix},
\end{align}
where $U_{n}^{pq}=[\varphi^{(p)}(x_n+\rho(l-k)-q]_{l,k\in \mathbb{Z}},~C_q=\{c_{\rho k+q}\}_{k\in \mathbb{Z}}^T~\text{and}~F_n^{(p)}=\{f^{(p)}(x_n+\rho l)\}_{l\in \mathbb{Z}}^T.$
Then the above system \eqref{Eqn2.4} can be written as 
\begin{align}
UC=F.
\end{align}
Notice that $U_{n}^{pq}$ is a Laurent operator with the symbol \begin{equation}\label{Psi1}
\Psi_{n}^{pq}(x)=\sum\limits_{k\in \mathbb{Z}}\varphi^{(p)}(x_n+\rho k-q)e^{2\pi \mathrm{i}kx},
\end{equation} 
The first condition in \eqref{decayproperty} ensures that the symbols $\Psi_n^{pq}$ are  continuous on the interval $[0,1]$ for all $n = 0,1,\dots, L-1$, $p=0,1,\dots,r-1,$ and $q=0,1,\dots,\rho-1$. Following the same argument as in \cite{Ghosh2023}, but replacing \textrm{sinc} generator by an arbitrary generator, it follows from Theorem \ref{IB} that the set $X$ is a CIS of order $r-1$ for $V(\varphi)$ if and only if 
$$\rho=L r \text{ and } \det \bm{\Psi}(x)\neq 0 \text{ for all } x\in[0,1],$$ 
where  $\bm{\Psi}(x)$, commonly known  as the polyphase matrix in signal processing, is defined by
\begin{align}\label{3.6}
\bm{\Psi}(x)=\begin{bmatrix}
\Psi_{0}^{00}(x) & \Psi_{0}^{01}(x) & \cdots& \Psi_{0}^{0,\rho-1}(x)\\
\Psi_{0}^{10}(x) & \Psi_{0}^{11}(x) & \cdots& \Psi_{0}^{1,\rho-1}(x)\\
\vdots & \vdots & \ddots & \vdots\\
\Psi_{0}^{r-1,0}(x) & \Psi_{0}^{r-1,1}(x) & \cdots& \Psi_{0}^{r-1,\rho-1}(x)\\
\hline
\vdots & \vdots & \ddots & \vdots\\
\hline
\Psi_{L-1}^{00}(x) & \Psi_{L-1}^{01}(x) & \cdots& \Psi_{L-1}^{0,\rho-1}(x)\\
\Psi_{L-1}^{10}(x) & \Psi_{L-1}^{11}(x) & \cdots& \Psi_{L-1}^{1,\rho-1}(x)\\
\vdots & \vdots & \ddots & \vdots\\
\Psi_{L-1}^{r-1,0}(x) & \Psi_{L-1}^{r-1,1}(x) & \cdots& \Psi_{L-1}^{r-1,\rho-1}(x)
\end{bmatrix}=\left[\Psi^{d_i, j}_{\alpha_i}(x)\right]_{i,j=0}^{\rho-1}, 
\end{align}
with $d_i= i \bmod r$ and $\alpha_i = \lfloor\frac{i}{r}\rfloor$.
Throughout this paper, $\lfloor x \rfloor$ and $\langle x\rangle$ denote the floor and fractional part of $x$, respectively.
The inverse of the matrix $U$ is given by $U^{-1}=[\widehat{\bm{\Psi}^{-1}}(i-j)]$ and hence every $f\in V(\varphi)$ can be written as
\begin{align}\label{eq 3.7}
f(t)=\sum_{n=0}^{L-1}\sum_{l\in\mathbb{Z}}\sum\limits_{i=0}^{r-1}f^{(i)}(x_n+ \rho l)\Theta_{ni}(t-\rho l),
\end{align}
where the interpolating (or synthesizing) kernels $\Theta_{ni}(t)$  are defined by
\begin{align}\label{shinp}
\Theta_{ni}(t)=\sum\limits_{\nu\in\mathbb{Z}}\sum\limits_{q=0}^{\rho-1}\widehat{(\bm{\Psi}^{-1})_{n}^{qi}}(\nu)\varphi(t-
\rho \nu-q),
\end{align}
for $n = 0, 1, \dots, L-1 $, $i=0,1,\dots,r-1.$ Here $\widehat{(\bm{\Psi}^{-1})_{n}^{qi}}(\nu)$ denotes the $(q, nr+i)$-th entry of $\widehat{\bm{\Psi}^{-1}}(\nu).$ 
Furthermore, the sampling bounds  are given by
$$A=\|\Phi\|^{-1}_0\sup\limits_{t\in[0,1]}\lambda_{min}({\bm{\Psi^{*}}(x)\bm{\Psi}(x))},~ B=\|\Phi\|_\infty\sup\limits_{t\in[0,1]}\lambda_{max}(\bm{\Psi}^{*}(x)\bm{\Psi}(x)),$$
with their derivation provided in the supplementary material.

By applying the Zak transform property \eqref{zak}, \eqref{Psi1} becomes 
\begin{align}{\label{psipsf}}
 \Psi_{n}^{pq}(x) &= \mathcal{Z}_{\rho}\varphi^{(p)}\left(x_n-q,\frac{-x}{\rho}\right)\nonumber \\
 &= \frac{1}{\rho} e^{\frac{-2 \pi \mathrm{i}}{\rho}(x_n-q)x} \mathcal{Z}_{\frac{1}{\rho}}  \widehat{\varphi^{(p)}}\left(\frac{-x}{\rho}, -x_n+q\right) \nonumber \\
 &= \frac{1}{\rho} e^{\frac{-2 \pi \mathrm{i}}{\rho}(x_n-q)x} \sum_{k \in \mathbb{Z}} \widehat{\varphi^{(p)}}\left( \frac{-x+ k}{\rho} \right) e^{\frac{2 \pi \mathrm{i}}{\rho}(x_n-q)k} \nonumber\\
&= \frac{1}{\rho} e^{-2\pi \mathrm{i} \left(\frac{x_n - q}{\rho}\right)x}\sum_{k \in \mathbb{Z}} \left(\frac{-2 \pi \mathrm{i}}{\rho}\right)^p (-x+k)^p \widehat{\varphi}\left( \frac{-x+ k}{\rho} \right) e^{2\pi \mathrm{i}\left( \frac{x_n - q}{\rho}\right)k} \nonumber\\
&=\frac{1}{\rho} e^{-2\pi \mathrm{i}  \left(\frac{x_n - q}{\rho}\right)x} \left(\tfrac{-2 \pi \mathrm{i}}{\rho}\right)^p\sum_{m \in \mathbb{Z}} \sum_{j=0}^{\rho-1}  (-x+\rho m+j)^p \widehat{\varphi}\left( \tfrac{-x+ \rho m+j}{\rho} \right)e^{2\pi \mathrm{i} \left(\frac{x_n-q}{\rho}\right) (\rho m +j)}\nonumber\\
&= \frac{1}{\rho} e^{-2\pi \mathrm{i}  \left( \frac{x_n - q}{\rho} \right)x} \left(\tfrac{-2 \pi i}{\rho}\right)^p  \sum_{j=0}^{\rho-1}  \sum_{m \in \mathbb{Z}} (-x+\rho m+j)^p\widehat{\varphi}\left(\tfrac{-x+ \rho m+j}{\rho}\right) e^{2\pi \mathrm{i} j \left(\frac{ x_n - q}{\rho} \right)} e^{2 \pi \mathrm{i} m x_n}. 
\end{align}
Let $y_j=\frac{-x +j}{\rho}$, $0 \leq j\leq \rho-1$. Then\\
$\left(\tfrac{-2 \pi \mathrm{i}}{\rho}\right)^p\sum_{m \in \mathbb{Z}}  (-x+\rho m+j)^p \widehat{\varphi}\left(\tfrac{-x+ \rho m+j}{\rho}\right)
 e^{2 \pi \mathrm{i} m x_n}$
\begin{align*} 
 &=\left(-2 \pi \mathrm{i}\right)^p \sum_{m \in \mathbb{Z}} 
 (y_j+m)^p\widehat{\varphi}\left(y_j+m\right) e^{2\pi \mathrm{i} m x_n} \nonumber\\
&=\mathcal{Z}{\widehat{\varphi^{(p)}}}\left(y_j, -x_n\right)
=e^{-2 \pi \mathrm{i} y_j x_n} \mathcal{Z}{\varphi^{(p)}}\left(x_n, y_j\right),
\end{align*}
from \eqref{zak}. Therefore, \eqref{psipsf} becomes
\begin{align}\label{eqn3.9}
\Psi_{n}^{pq}(x) 
&=\frac{1}{\rho}  \sum_{j=0}^{\rho-1} 
\mathcal{Z}{\varphi^{(p)}}\left(x_n, \frac{-x+j}{\rho}\right)
e^{\frac{2\pi \mathrm{i} (x_n-q)}{\rho} (-x + j)} e^{\frac{-2\pi \mathrm{i} x_n}{\rho} (-x + j)}\nonumber \\
&= \frac{1}{\rho}  \sum_{j=0}^{\rho-1} 
\mathcal{Z}{\varphi^{(p)}}\left(x_n, \frac{-x+j}{\rho}\right)
e^{\frac{-2\pi \mathrm{i} q}{\rho} (x - j)}.
\end{align}
Let $u_j(t;x)=\mathcal{Z}\varphi\left(t, \frac{-x+j}{\rho}\right),$ $j=0,1, \dots,\rho-1. $ Then from the definition of the Zak transform, we have 
$$D^{p}u_j(t;x)=\mathcal{Z}\varphi^{(p)}\left(t, \frac{-x+j}{\rho}\right), \text{ where }~~D^p\equiv\frac{\partial^p}{\partial t^p}.$$ 
Hence we can rewrite \eqref{eqn3.9} as 
$$\bm\Psi(x)= \frac{1}{\rho} G(x)B(x),$$
where $$G(x)=\left[D^{d_i}u_j(t_i;x)\right]_{i,j=0}^{\rho-1} \text{ with } d_i=i\bmod r \text{ and } t_i=x_{\lfloor\frac{i}{r}\rfloor};$$ $$B(x)=\left[e^{-2 \pi \mathrm{i} q\frac{(x-j)}{\rho}}\right]_{q,j=0}^{\rho-1}.$$
It is easy to verify that $\det B(x)\neq 0$ for all $x\in[0,1]$. 
Consequently, we obtain the following result.
\begin{thm}
Let $\varphi\in \mathcal{A}_r$ be a  stable generator for $V(\varphi)$ and let
$X=\{x_n+\rho \ell:~0\leq n\leq L-1,~l\in\mathbb{Z}\}$ be a PNS set of period $\rho$ with length $L$.
Define 
$$d_i=i\bmod r,t_i=x_{\lfloor\frac{i}{r}\rfloor} \text{ and }u_j(t;x)=\mathcal{Z}\varphi\left(t, \frac{-x+j}{\rho}\right),~~j=0,1, \dots,\rho-1. $$ 
If $\rho=Lr$, then $X$  is a CIS of order $r-1$ for $V(\varphi)$ if and only if
$$\det \left[D^{d_i}u_j(t_i;x)\right]_{i,j=0}^{\rho-1}\neq 0 \text{ for all x }\in [0,1], \text{ where } D^p\equiv\dfrac{\partial^p}{\partial t^p}.$$
\end{thm}
We now turn our attention to the perfect reconstruction formula. To achieve perfect reconstruction, it is essential that both the generator and the interpolating kernels have compact support. Our goal is, therefore, to construct interpolating kernels with compact support. A general method for constructing such kernels, specifically when the generator is a B-spline, was presented in \cite{BBSV2, Butzer5}.

We now assume that $\varphi$ is  a $(r-1)$-times differentiable function  with support $[0,\mu]$ and we choose the offset vector $x=(x_0, x_1, \dots,x_{L-1})$
satisfying 
$$\alpha\leq x_0 < x_1<\cdots < x_{L-1}<\beta,$$
where $\alpha \geq 0$ and $\beta \leq \rho $. Then $\Psi_n^{pq}(x)$ is non-zero only if there exists an integer $k \in \mathbb{Z}$ such that 
$$0 < x_n + \rho k - q < \mu.$$
By our assumptions on the PNS pattern, this inequality implies that
$$-1< k < 1+\frac{\mu-1}{\rho}.$$
When $\rho \geq \mu-1$, the only terms with $k=0$ and $k=1$ contribute to the summation in \eqref{Psi1}.
Consequently, the functions $\Psi_{n}^{pq}(x)$ are of the form
\begin{equation}\label{Psi2}
\Psi_{n}^{pq}(x)=\varphi^{(p)}(x_n - q)+\varphi^{(p)}(x_n + \rho - q) e^{2 \pi \mathrm{i} x},
\end{equation} for all $0 \leq p \leq r -1,$ $0 \leq q \leq \rho-1$, and $0 \leq n \leq L-1.$

When $\alpha=\beta-1=s$ for some $s\in\left\{0,1,\dots, \rho -1\right\}$ 
and $\rho \geq \mu$,  the matrix $\bm{\Psi}(x)$  has the following structure:
\begin{equation} \label{PSIPR}
\resizebox{\textwidth}{!}{$
\bm{\Psi}(x) = 
\begin{bmatrix}
\varphi^{(0)}(x_0) & \cdots & \varphi^{(0)}(x_0-s) & \varphi^{(0)}(x_0+ \rho-s-1)z & \cdots & \varphi^{(0)}(x_0+1)z \\
\varphi^{(1)}(x_0) & \cdots & \varphi^{(1)}(x_0-s) & \varphi^{(1)}(x_0+ \rho -s-1)z & \cdots & \varphi^{(1)}(x_0+1)z \\
\vdots & \ddots & \vdots & \vdots & \ddots & \vdots \\
\varphi^{(r-1)}(x_0) & \cdots & \varphi^{(r-1)}(x_0-s) & \varphi^{(r-1)}(x_0+ \rho -s-1)z & \cdots & \varphi^{(r-1)}(x_0+1)z \\ 
\hline
\vdots & \ddots & \vdots & \vdots & \ddots & \vdots \\ 
\hline
\varphi^{(0)}(x_{L-1}) & \cdots & \varphi^{(0)}(x_{L-1}-s) & \varphi^{(0)}(x_{L-1}+ \rho -s-1)z & \cdots & \varphi^{(0)}(x_{L-1}+1)z \\ 
\varphi^{(1)}(x_{L-1}) & \cdots & \varphi^{(1)}(x_{L-1}-s) & \varphi^{(1)}(x_{L-1}+ \rho -s-1)z & \cdots & \varphi^{(1)}(x_{L-1}+1)z \\
\vdots & \ddots & \vdots & \vdots & \ddots & \vdots \\
\varphi^{(r-1)}(x_{L-1}) & \cdots & \varphi^{(r-1)}(x_{L-1}-s) & \varphi^{(r-1)}(x_{L-1}+ \rho -s-1)z & \cdots & \varphi^{(r-1)}(x_{L-1}+1)z
\end{bmatrix}_{\rho  \times \rho}
$.}
\end{equation}
Here and throughout this paper, we use the notation  $z=e^{2 \pi \mathrm{i} x}$. Applying elementary column operations to
$\bm{\Psi}(x)$, we get
$$\det \bm{\Psi}(x)= (-1)^{\rho-1} z^{\rho-s-1} \det \left[  \varphi_j^{(d_i)}(t_i) \right]_{i,j=0}^{\rho-1},
$$
where
$$t_i=x_{\lfloor\frac{i}{r}\rfloor} \text{ and } \varphi_j(x)= \varphi(x-(s+1-\rho+j)),~  0 \leq i,j \leq \rho-1.$$
Therefore, $\bm{\Psi}(x)$ is invertible for all $x\in [0,1]$  if and only if
$$\det \left[  \varphi_j^{(d_i)}(t_i) \right]_{i,j=0}^{\rho-1}
\neq 0.$$
In this case, it is clear from \eqref{PSIPR} that the inverse of $\bm{\Psi}(x)$ is of the form 
$$\bm{\Psi}^{-1}(x)=A+Bz^{-1},$$ 
where $A$ and $B$ are some $\rho\times \rho$ matrices such that the last $\rho-s-1$ rows of $A$ and the first $s+1$ rows of $B$ are all zero rows.  Therefore, the Fourier coefficients  $\widehat{\bm{\Psi}^{-1}}(\nu)$  are the zero matrix for every $\nu\neq -1,0.$ Furthermore, we have
$$\widehat{(\Psi^{-1})_{n}^{qi}}(-1)=[\widehat{\Psi^{-1}}(-1)]_{q,nr+i} =0, \text{ for }q=0,1,\dots,s, $$
and
$$\widehat{(\Psi^{-1})_{n}^{qi}}(0)=[\widehat{\Psi^{-1}}(0)]_{q,nr+i} =0, \text{ for } q=s+1,\dots,\rho-1.$$ 
Consequently, \eqref{shinp} becomes
\begin{align}\label{eqn3.12}
\Theta_{ni}(t)&=\sum\limits_{\nu=-1}^{0}\sum\limits_{q=0}^{\rho-1}\widehat{(\Psi^{-1})_{n}^{qi}}(\nu)\varphi(t-\rho \nu-q)\nonumber\\
&=\sum\limits_{\nu=-1}^{0}\sum\limits_{q=0}^{s}\widehat{(\Psi^{-1})_{n}^{qi}}(\nu)\varphi(t-\rho \nu-q)+\sum\limits_{\nu=-1}^{0}\sum\limits_{q=s+1}^{\rho-1}\widehat{(\Psi^{-1})_{n}^{qi}}(\nu)\varphi(t-\rho \nu-q)\nonumber \\
&=\sum\limits_{q=0}^{s}\widehat{(\Psi^{-1})_{n}^{qi}}(0)\varphi(t-q)+\sum\limits_{q=s+1}^{\rho-1}\widehat{(\Psi^{-1})_{n}^{qi}}(-1)\varphi(t+\rho-q).
\end{align}
It is clear that the interpolating kernels $\Theta_{ni} \in \mathcal{A}_r$. Since the support of $\varphi$ is $[0,\mu]$, we obtain from \eqref{eqn3.12} that  
\begin{align*}
\text{ the support of }\Theta_{ni}
&= \bigcup\limits_{q=0}^{s} \left[q,\mu+q\right] \cup \bigcup\limits_{q=s+1}^{\rho-1}\left[-\rho +q, \mu-\rho+q\right]\nonumber\\
&=\left[0,\mu+s\right] \cup \left[-\rho+s+1, \mu-1\right]\nonumber\\
&=\left[-\rho+s+1,\mu+s\right].
\end{align*}
We summarize our discussion in the following result.
\begin{thm}\label{theorem3.2}
Let $\varphi \in \mathcal{A}_r$ be a stable generator for $V(\varphi)$
with support $[0,\mu]$. Let $X = \big\{x_n+\rho l:0\leq n\leq L-1,~l\in\mathbb{Z}\big\}$, where $x_0,x_1,\dots, x_{L-1} \in [s,s+1)$ for some $s\in  \{0,1,\dots,\rho-1\}$. Define
$$t_i=x_{\lfloor\frac{i}{r}\rfloor} \text{ and } \varphi_j(x)= \varphi(x-(s+1-\rho+j)),~  0 \leq i,j \leq \rho-1.$$
If $\rho=Lr\geq \mu $, then $X$ is a CIS of order $r-1$ for $V(\varphi)$ if and only if 
$$\det\left[  \varphi_j^{(d_i)}(t_i) \right]_{i,j=0}^{\rho-1}\neq 0.$$ 
In this case, every $f\in V(\varphi)$ can be written as
\begin{align}\label{sampling formula}
f(t)=\sum_{n=0}^{L-1}\sum_{l\in\mathbb{Z}}\sum\limits_{i=0}^{r-1}f^{(i)}(x_n+ \rho l)\Theta_{ni}(t-\rho l),
\end{align}
where the interpolating kernels $\Theta_{ni}$, as defined in  \eqref{eqn3.12}, belong to $\mathcal{A}_r$ and have  support on the interval $\left[-\rho+s+1,\mu+s\right].$ 
\end{thm}
We now use our main result to explore the complete interpolation property of certain sampling sets for shift-invariant spaces generated by B-splines. B-splines are highly valuable in practical applications due to the availability of efficient and robust algorithms for their computation.

Let $\left\{y_i:i\in\mathbb{Z}\right\}$ be a sequence of real numbers such that $y_i\leq y_{i+1}$ for all $i$. Given integers $i$ and $m>0$, the $m$-th order (normalized) B-splines associated with the knots $y_i\leq\cdots\leq y_{i+m}$ is defined by
\begin{equation}
N_i^m(t)=
\begin{cases}
(-1)^m(y_{i+m}-y_i)[y_i, \dots,y_{i+m}](t-y)_+^{m-1}, & \text{ if } y_i< y_{i+m}\\
 0, & \text{otherwise},
\end{cases}
\end{equation}
for all real $t$, where $[y_i, \dots, y_{i+m}]$ is a divided difference and 
$u_+ := \max\{u, 0\}.$

Let $Q_m(t):= N_0^m(t)$ denote the B-spline associated with the simple knots $0,1,\dots,m$. Then the normalized B-splines associated with the knots 
$i\leq i+1\leq\cdots\leq i+m$ is given by 
$$N^m_i(t)=Q_{m}(t-i).$$

The B-spline $Q_m$ is $(m-2)$-times differentiable with support on the interval $\left[0,m\right]$. It is symmetric about the line $t=\frac{m}{2}$. The B-splines and their derivatives can be computed by the following formulae
$$Q_m(t)=\dfrac{1}{(m-1)!}\displaystyle\sum_{j=0}^{m}(-1)^j
\binom{m}{j} \left(t-j\right)_+^{m-1},~m\geq 2,$$
and 
$$\hspace{-1.55cm} Q_m^\prime(t)=Q_{m-1}(t)-Q_{m-1}(t-1),~m >2,$$
respectively. It is well known that the B-spline $Q_m$ is an $(m-2)$-regular stable generator for $V(Q_m)$. The following theorem, known as the Schoenberg-Whitney conditions, provides a precise criterion for the existence and uniqueness of a solution to a Hermite interpolation problem in spline spaces.

\begin{thm}\cite{Schumaker1981} \label{Spline}
Let $N_1^m, \dots, N_n^m$ be a set of B-splines of order $m$ associated with the knots $y_1\leq\cdots\leq y_{n+m}$.
Let $t_1 \leq \cdots \leq t_n$ with $t_i<t_{i+m}$, all $i$. Then 
 $$ \det \left[ D_+^{d_i} N_j^m (t_i) \right]_{i,j=1}^{n}\geq0,$$
and strict positivity holds if and only if
$$
t_i \in  
\begin{cases}  
[y_i, y_{i+m}), & \text{if } d_i \geq m - \alpha_i, \\  
(y_i, y_{i+m}), & \text{otherwise },  
\end{cases}  
$$
where $d_i=\max\left\{j : t_i=\cdots = t_{i-j}\right\}$ and $\alpha_i=\max\left\{j : y_i = \cdots = y_{i+j-1}\right\}$,  $i = 1,2, \dots, n,$ and $D_+$ denotes the right hand derivative operator.
\end{thm}
\begin{cor}\label{cor 3.1}
Assume that $\rho=Lr=m>r$. If $x_0,x_1,\dots, x_{L-1} \in [s,s+1)$ for some $s\in  \{0,1,\dots,\rho-1\}$, then $X = \big\{x_n+\rho l:0\leq n\leq L-1,~l\in\mathbb{Z}\big\}$ is a CIS of order $r-1$ for $V(Q_m)$. 
\end{cor}
\begin{proof}
If we choose $\varphi(t)=Q_m(t)$, then
$$\varphi_j(t)=Q_m(t-(s+1-\rho+j))=N_j^m(t),$$ 
where $N_j^m(t)$ is the B-spline of order $m$ associated with knots $y_j<y_{j+1}<\cdots<y_{j+m}$ with $y_j=s+1-\rho+j$, $j=0,1,\dots, \rho-1$. 
Now, it is clear that 
$$ \varphi^{(d_i)}_j (t_i)= D^{d_i}N_j^{m} (t_i),~~ d_i= i \bmod r,  ~~t_i=x_{\lfloor\frac{i}{r}\rfloor}, ~~\alpha_i=0, ~~0 \leq i,j \leq \rho-1. $$
Now it follows from Theorem \ref{Spline} that
$\det\left[  \varphi^{(d_i)}_j(t_i) \right]_{i,j=0}^{\rho-1}\neq 0$
if and only if
\begin{equation}\label{condition for ti}
 t_i \in \left(s+1-\rho +i, s+1-\rho +m+i\right) \text{ for } i=0,1, \dots, \rho-1.
\end{equation}
Since $x_j=t_{jr+k}$, $k=0,1, \dots, r-1$, \eqref{condition for ti} equivalent to
\begin{equation}\label{condition fot xi}
x_j \in \left(s-\rho +(j+1)r, s-\rho+jr+m+1\right), \text{ for }j=0,1,\dots, L-1.
\end{equation}
If $x_0,x_1,\dots, x_{L-1} \in [s,s+1)$ for some $s\in \{0,1,\dots,\rho-1\}$, \eqref{condition fot xi} is always true. Now our desired result follows from Theorem \ref{theorem3.2}.
\end{proof}
The following result was originally proved in \cite{Priyanka}. We present a simplified proof using the Schoenberg–Whitney conditions.
\begin{cor}\label{cor 3.2}
For $m \geq 2$, $\frac{1}{2}+(m-1)\mathbb{Z}$ is not a CIS of order $m-2$ for $V(Q_m)$, but $(m-1)\mathbb{Z}$ is.
\end{cor}
\begin{proof}
For the sample set $(m-1)\mathbb{Z}$,  we have
\begin{equation*}\label{Psi3}
\Psi_{0}^{pq}(x)=Q_m^{(p)}(m-1 - q) e^{2 \pi \mathrm{i} x},~~ 0 \leq p,q \leq m-2,
\end{equation*}
from \eqref{Psi2}. Consequently, the determinant of $\bm{\Psi}(x)$ is given by
$$\det \bm{\Psi}(x)=  z^{m-1} \det \left[ Q_m^{(i)}(m-1-j)\right] =z^{m-1} \det \left[D^{i}N_j^{m} (0)\right]_{i,j=0}^{m-2},
$$
where
$N_j^m(t)$ is the  B-spline of order $m$ associated with the knots $y_j<y_{j+1}<\cdots<y_{j+m}$, given by $y_j=j-m+1$, $j=0,1,\dots, m-2$.  
Now it follows from Theorem \ref{Spline} that
$\det\left[ D^{i} N^m_j(0) \right]_{i,j=0}^{m-2}\neq 0$
if and only if
\begin{equation*}
 0 \in \left(i-m+1, i+1\right) \text{ for } i=0,1, \dots, m-2.
\end{equation*}
This condition always holds. Hence $(m-1)\mathbb{Z}$ is a CIS of order $m-2$ for $V(Q_m)$. 
For the sampling set $\frac{1}{2}+(m-1)\mathbb{Z}$,  the symmetric property of the B-spline $Q_m$ implies that
\begin{align}
\Psi_{0}^{p0}(x)
=Q_m^{(p)}\Big(\frac{1}{2}\Big)+Q_m^{(p)}\Big(m-\frac{1}{2}\Big) e^{2 \pi i x} = Q_m^{(p)}\Big(\frac{1}{2}\Big)(1+ e^{2 \pi \mathrm{i} x}),\nonumber
\end{align}
for all $0\leq p \leq m-2$. Clearly, $\det \bm{\Psi}(0)=0$, which confirms that $\frac{1}{2}+(m-1)\mathbb{Z}$ is not a CIS of order $m-2$ for $V(Q_m)$. 
\end{proof}

\begin{rem}
Consider  the sampling set $a + r\mathbb{Z}$ for the space $V(Q_m)$ and $m>r$. In this context, one can show that
$$\det \bm\Psi(z) = z^{r - 2a} P(a, z),$$
where
$$P(a, z) = c_0 + c_1 z + \cdots + c_{m - r + 2a - 1} z^{m - r + 2a - 1}
$$
is a self-inversive polynomial of degree $m - r + 2a - 1$.
Owing to the complexity of the calculations, the detailed proof of the self-inversive property is provided in the supplementary material; interested readers can refer to it for a full derivation.

It then follows that if $a + r\mathbb{Z}$ forms a CIS of order $r-1$ for $V(Q_m)$, we must have either
$$
\text{(i) } m \text{ even and } a = \left\langle \frac{r+1}{2} \right\rangle, \quad \text{or} \quad
\text{(ii) } m \text{ odd and } a = \left\langle \frac{r}{2} \right\rangle.
$$
It has been conjectured in \cite{Priyanka} that the converse of above statement is true. We further conjecture that all zeros of the polynomial $P(a,z)$ are simple and lie on the positive real axis. If this property were established, the converse statement would follow immediately from the properties of self-inversive polynomials. Despite considerable effort, this result has not yet been proven.
\end{rem}

We now present an example illustrating that a PNS set can be a CIS for one shift-invariant space but not for another. Furthermore, the associated interpolating kernels are not compactly supported.

Consider the sampling set
$$X=(0.5+4\mathbb{Z}) \cup (2.5+4\mathbb{Z}) 
\text{ with } L=2, ~\rho=4, \text{ and } r=2.$$
When the   generator is $Q_3$,  the associated polyphase matrix $\bm{\Psi}(x)$ is given by
\begin{equation*}
    \renewcommand{\arraystretch}{1.4}
\bm{\Psi}(x) = 
\begin{bmatrix}
\frac{1}{8} & 0 & \frac{z}{8} & \frac{3z}{4} \\ 
\frac{1}{2}& 0 & -\frac{z}{2} &0 \\ 
\frac{1}{8} & \frac{3}{4} & \frac{1}{8} & 0 \\
 -\frac{1}{2} & 0 & \frac{1}{2} & 0 
\end{bmatrix}.
\end{equation*}
Its determinant is  $\det\bm{\Psi}(x)=\frac{9}{64}z(z-1)$, which vanishes at $z=1$. Therefore, $X$ is not a CIS of order $1$ for the space $V(Q_3)$.
However, when the generator is $Q_4$, the polyphase matrix $\bm{\Psi}(x)$ becomes
\begin{equation*}
 \renewcommand{\arraystretch}{1.4} 
 \bm{\Psi}(x) = 
\begin{bmatrix}
\frac{1}{48}& \frac{z}{48} &\frac{23z}{48} & \frac{23z}{48} \\ 
\frac{1}{8} &-\frac{z}{8} & -\frac{5z}{8} &\frac{5z}{8} \\ 
\frac{23}{48} &\frac{23}{48}  & \frac{1}{48}  & \frac{z}{48} \\ 
-\frac{5}{8} & \frac{5}{8} & \frac{1}{8}  & -\frac{z}{8}
\end{bmatrix},
\end{equation*}
with $\det\bm{\Psi}(x)=-\frac{z \left(9 z^2 - 1426 z + 9\right)}{4096}\neq 0
$ for all $x\in[0,1]$. Hence, $X$ is a CIS of order $1$ for the space $V(Q_4)$.
Since the determinant of $\bm{\Psi}(x)$ is not a monomial, the corresponding interpolating kernels are not compactly supported.

\section{Signal Prediction from Past Derivative Samples}
Let $1\leq p<\infty$ and let $\varphi \in \mathcal{A}_r$ be a stable generator for $V(\varphi)$. Then there exist two positive constants $m_p$ and $M_p$ such that
$$m_p\|c\|_p\leq \left\|\sum\limits_{k\in\mathbb{Z}}c_k\varphi(\cdot-k)\right\|_p\leq M_p\|c\|_p,$$
for every $c=(c_k)\in \ell^p(\mathbb{Z})$.
Consequently, the space
$$V^p(\varphi):=\left\{f\in L^p(\mathbb{R}):f(\cdot)=\sum_{k\in\mathbb{Z}}c_k\phi(\cdot-k), (c_k)\in \ell^p(\mathbb{Z})\right\}$$ is a closed subspace of $L^p(\mathbb{R})$; 
see \cite{AlGr}.
Assume that $X= \big\{x_n+\rho l:0\leq n\leq L-1,~l\in\mathbb{Z}\big\}$ is a CIS of order $r-1$ for $V(\varphi)$.
Then, following the argument in \cite{GhAn}, it can be shown that
every $f \in V^p(\varphi)$ can be written as
\begin{align*}
f(t)=\sum_{n=0}^{L-1}\sum_{l\in\mathbb{Z}}\sum\limits_{i=0}^{r-1}f^{(i)}(x_n+\rho l)\Theta_{ni}(t-\rho l),
\end{align*}
where the interpolating kernels $\Theta_{ni}$ are defined as in \eqref{Psi2}.
This sampling formula motivates the introduction of
the sampling series
\begin{equation}\label{sampling series}
(S_Wf)(t):= \sum\limits_{l\in\mathbb{Z}}\sum\limits_{i=0}^{r-1}\sum\limits_{n=0}^{L-1}\frac{1}{W^i}f^{(i)}\left(\frac{x_n+\rho l}{W}\right)\Theta_{ni}(Wt-\rho l), 
\end{equation}
where $f$ belongs to a suitable class of real or complex valued functions. 

In order to explore the rate of convergence of the sampling series $S_Wf$ to its approximation of $f$ we analyze its reproducing kernel property, which plays a key role in determining the accuracy of the approximation. The sampling operator $S_Wf$ is said to satisfy the \textit{reproducing polynomial property} of order $\kappa-1$ if $$S_Wf(t)=f(t), \text{ for all polynomials $f$ of degree } \leq \kappa-1.$$ The following lemma provides a necessary and sufficient condition for this property,  with its proof given in the Appendix.

\begin{lem}\label{lemma4.1}
Let $\varphi\in \mathcal{A}_r$ be a stable generator for $V(\varphi)$ and $X=\{x_n+\rho \ell:~0\leq n\leq L-1,~l\in\mathbb{Z}\}$ be a CIS of order $r-1$ for $V(\varphi).$
Then the following statements are equivalent.
\begin{itemize}
\item[$(i)$] The sampling operator $S_W$ defined in \eqref{sampling series} satisfies the reproducing polynomial property of order $\kappa-1.$
\item[$(ii)$] For $j=0,1,\dots,\kappa-1,$
\begin{equation}\label{vanishingmoment1}
\sum\limits_{i=0}^{r-1}\tbinom{j}{i} i!\sum\limits_{l\in\mathbb{Z}}\sum\limits_{n=0}^{L-1}\left(x_n+\rho l-t\right)^{j-i}\Theta_{ni}(t-\rho l)=\delta_{j0},~~ t\in \mathbb{R}.
\end{equation}
\item [$(iii)$] For $j=0,1,\dots,\kappa-1$ and $l\in\mathbb{Z},$
\begin{equation}\label{vanishingmoment2}
\sum_{i=0}^{r-1} \binom{j}{i} \, i! \, \frac{1}{(2\pi \mathrm{i})^{\, j-i}}\;
F_{li}^{(j-i)}(l/\rho)
=\rho \, \delta_{l0}\, \delta_{j0},
\end{equation}
where $F_{li}(w)=\sum\limits_{n=0}^{L-1}
e^{\, 2\pi \mathrm{i} x_n \left(w - \frac{l}{\rho}\right)}
\, \widehat{\Theta}_{n i}(w).$

\end{itemize}
\end{lem}
Equation~\eqref{vanishingmoment1} reduces to the classical vanishing moment conditions discussed in \cite{BBSV2} when $r=1$ and $L=1$, and to the conditions discussed in \cite{Priyanka} when $L=1$. Therefore, we refer to \eqref{vanishingmoment1} as the generalized vanishing moment conditions for the interpolating kernels $\Theta_{ni}$.

We now assume that $\varphi\in \mathcal{A}_r$ is a  stable generator for $V(\varphi)$ with support $[0,\mu]$. Additionally, let
$$x_0,x_1,\dots, x_{L-1} \in [s,s+1) \text{ for some } s\in  \{0,1,\dots,\rho-1\} \text{ and } \rho=Lr\geq \mu.$$
Recall that if $X= \big\{x_n+\rho l:0\leq n\leq L-1,~l\in\mathbb{Z}\big\}$ is a CIS of order $r-1$ for $V(\varphi)$,  then the interpolating kernels $\Theta_{ni}$ in \eqref{sampling series} are compactly supported on the interval $\left[-\rho+s+1,\mu+s\right].$  
To determine the rate of approximation of the sampling series $S_Wf$ when the interpolating kernels have compact support,  we utilize \eqref{vanishingmoment1} along with a Taylor expansion. Once this is established, we apply our interpolation theorem to prove the following result. For the ease of presentation of the paper, a complete proof is provided in the Appendix.
\begin{thm}\label{approximation operator}
Let $\varphi \in \mathcal{A}_r$ be a stable generator for $V(\varphi)$ with support $[0,\mu]$. Additionally, let
$$x_0,x_1,\dots, x_{L-1} \in [s,s+1) \text{ for some } s\in  \{0,1,\dots,\rho-1\} \text{ and } \rho=Lr \geq \mu.$$
and $X=\{x_n+\rho\ell:~0\leq n\leq L-1,~l\in\mathbb{Z}\}$ be a CIS of order $r-1$ for $V(\varphi).$ If $S_W$ satisfies the reproducing polynomial property of order $r-1,$ then for each $f\in\Lambda^p_{r}$, there holds the error estimate 
\begin{equation}
\|S_{W}f-f\|_p\leq \sum\limits_{i=0}^{r-1}c_i\tau_r\left(f^{(i)};\dfrac{1}{W}\right)_p, ~ \text{for every}~ W>0,
\end{equation}
where the constants $c_i,$ $i=0,1,\dots, r-1$  depend only on $K_1$ and $K_2.$ 
In addition, if  $\tau_r(f^{(i)},W^{-1})_p=\mathcal{O}(W^{-\alpha})$ as $W\to \infty$ for each $i=0,1,\dots,r-1,$ then
\begin{equation}
\|S_Wf-f\|_p= \mathcal{O}(W^{-\alpha}) \quad (W\to\infty).
\end{equation}
\end{thm}

In order to predict a signal by the sampling series \eqref{sampling series}, we need to assume that the generator $\varphi$ should be compactly supported in the interval $(0,\infty)$. In the following, we will construct compactly supported  interpolating kernels having support $(0,\infty)$ that satisfy the generalized vanishing moment condition of order $r-1$ using the perfect reconstruction formula and hence we predict a signal from finite number of sample points. The following lemma provides a key framework for this construction.
\begin{lem}\label{lemma4.2}
Consider the Vandermonde system
\begin{equation}\label{Vandermonde}
\sum_{p=0}^{\rho-1} a_p (-\epsilon_p)^j = \delta_{j0}, 
\qquad j = 0,1,\dots,\rho-1,    
\end{equation}
where \(\epsilon_0 < \epsilon_1 < \dots < \epsilon_{\rho-1}\) are given real numbers. 
The solution $(a_p)_{p=0}^{\rho-1}$ is unique and coincides with the Lagrange interpolation weights at \(x=0\):
\begin{equation}\label{a_p2}
  a_p = \prod_{\substack{q=0 \\ q\neq p}}^{\rho-1} \frac{\epsilon_q}{\epsilon_q - \epsilon_p}, \quad p=0,1,\dots,\rho-1.  
\end{equation}
\end{lem}
\begin{proof}
Define the Lagrange basis for the nodes $x_p=-\epsilon_p$:
$$
\ell_p(x)=\prod_{\substack{q=0 \\q\neq p}}^{\rho-1}\frac{x-x_q}{x_p-x_q}, \quad p=0,1,\dots,\rho-1.
$$ 
Evaluating at $x=0$ gives
$\ell_p(0)=a_p.$ For any $j=0,\dots,\rho-1$, the polynomial $f(x)=x^j$  satisfies the exact Lagrange interpolation formula:
$$
f(x)=\sum_{p=0}^{\rho-1} f(x_p)\,\ell_p(x).
$$
Evaluating at $x=0$ yields \eqref{Vandermonde}.
Since the Vandermonde matrix is invertible, this solution is unique.
\end{proof}
\begin{rem}
If the nodes are equally spaced, $\epsilon_p = \epsilon_0 + p\,d$ with $d_0 = \epsilon_0/d$, 
the weights can be expressed explicitly in terms of the Gamma function:
\begin{equation}
a_p = \frac{(-1)^p}{p!\,(\rho-1-p)!\,(d_0 + p)} \, \frac{\Gamma(d_0 + \rho)}{\Gamma(d_0)}, 
\qquad p = 0,1,\dots,\rho-1,
\end{equation}
where $\Gamma(\cdot)$ denotes the Gamma function.
\end{rem}
We now introduce the modified interpolating kernels defined by
\begin{equation}
\widetilde{\Theta}_{ni}(t) = \sum_{p=0}^{\rho-1} a_p\, \Theta_{ni}(t - \epsilon_p), 
\quad 0 \le n \le L-1, \; 0 \le i \le r-1,
\end{equation}
where the functions $\Theta_{ni}$ are as in \eqref{eqn3.12} and the coefficients $a_p$ are precisely those given by Lemma~\ref{lemma4.2}. 
It can be shown easily that the support of $\widetilde{\Theta}_{ni}$ is
$$
\operatorname{supp}(\widetilde{\Theta}_{ni}) = \bigl[ -\rho + s + 1 + \epsilon_0,\ \mu + s + \epsilon_{\rho-1} \bigr].
$$
To guarantee that $\widetilde{\Theta}_{ni}$ is supported entirely in $(0,\infty)$, we assume $\epsilon_0 \ge \rho$.

Let us introduce the operator $(\widetilde{S}_Wf)$ which is defined as
\begin{align}\label{eq4.6}
(\widetilde{S}_Wf)(t):= \sum\limits_{l\in \mathbb{Z}}\sum\limits_{i=0}^{r-1}\sum\limits_{n=0}^{L-1}\frac{1}{W^i}f^{(i)}\left(\frac{x_n+\rho l}{W}\right)\widetilde{\Theta}_{ni}(Wt-\rho l).
\end{align}
Since the support of $\widetilde{\Theta}_{ni}$ is
$\left[ -\rho + s + 1 + \epsilon_0,\ \mu + s + \epsilon_{\rho-1} \right],$
the summation extends only over those values of $l \in \mathbb{Z}$ for which
$$-\rho+s+1+\epsilon_0 \leq Wt-\rho l \leq \mu+s+\epsilon_{\rho-1}.$$
This condition is equivalent to restricting $l$ to the set 
$$\Omega_t:=\left\{l \in \mathbb{Z} : t-\frac{\mu+s+\epsilon_{\rho-1}-x_n}{W}\leq \frac{x_n+\rho l}{W}\leq t- \frac{-\rho+s+1+\epsilon_0-x_n}{W}\right\}.$$ Since $-\rho+s+1+\epsilon_0-x_n > 0$ and the number of elements in $\Omega_t$ is at most $2+\left\lfloor \frac{\mu-1+\epsilon_{p-1}-\epsilon_0}{\rho}\right\rfloor$, evaluating the sampling series at any given time requires at most  $\rho\left(2+\left\lfloor \frac{\mu-1+\epsilon_{p-1}-\epsilon_0}{\rho}\right\rfloor\right)$ samples from the past.  Thus the series \eqref{eq4.6} takes the form
\begin{equation}\label{prediction operator}
(\widetilde{S}_Wf)(t)= \sum\limits_{l\in \Omega_t}\sum\limits_{i=0}^{r-1}\sum\limits_{n=0}^{L-1}\frac{1}{W^i}f^{(i)}\left(\frac{x_n+\rho l}{W}\right)\widetilde{\Theta}_{ni}(Wt-\rho l),
\end{equation} which is referred to as the prediction operator.  The operator $(\widetilde{S}_Wf)(t)$ uses past values of the function $f$, its derivatives, and a series of time shifts to produce an estimate of the function at time $t$.

If $S_W$ satisfies the reproducing polynomial property of order $r-1$, then we show that the prediction operator $\widetilde{S}_W$  also satisfies the reproducing polynomial property of the same order. In order to prove this, it is enough to  show that $\widetilde{S}_Wf$ satisfies the condition $(iii)$ of Lemma \ref{lemma4.1} for $\widetilde{\Theta}_{ni}$.
Observe  that the Fourier transform of $\widetilde{\Theta}_{ni}(t)$ can be expressed as $$\widehat{\widetilde{\Theta}}_{ni}(w)= P(w)\widehat{\Theta}_{ni}(w),$$
where $P(w)=\sum\limits_{p=0}^{\rho-1} a_p e^{-2 \pi \mathrm{i}\epsilon_p w}.$ It is clear from \eqref{Vandermonde} that $P^{(j)}(0)=(2 \pi \mathrm{i})^j \delta_{j0}.$

Let us define $$\widetilde{F}_{li}(w):=\sum\limits_{n=0}^{L-1}
e^{\, 2\pi \mathrm{i} x_n \left(w - \frac{l}{\rho}\right)}
\, \widehat{\widetilde{\Theta}}_{n i}(w)=P(w)F_{li}(w).$$
Using Leibniz rule for derivatives and changing the order of summation, we have
$$\sum\limits_{i=0}^{r-1} \binom{j}{i} \, i! \, \frac{1}{(2\pi \mathrm{i})^{\, j-i}}
\;
\widetilde{F}_i^{(j-i)}(l/\rho)$$
\begin{align}\label{eqn1.9}
    &=\sum_{i=0}^{r-1} \binom{j}{i} \, i! \, \frac{1}{(2\pi \mathrm{i})^{\, j-i}}\, \sum_{m=0}^{j-i} \binom{j-i}{m} P^{(m)}(l/\rho) F_{li}^{(j-i-m)}(l/\rho)\nonumber \\
&=\sum_{m=0}^{j} \binom{j}{m} \, \frac{1}{(2\pi \mathrm{i})^{\, m}}\; P^{(m)}(l/\rho) \sum_{i=0}^{r-1} \binom{j-m}{i}  \, i! \frac{1}{(2\pi \mathrm{i})^{\, j-i-m}} F_{li}^{(j-i-m)}(l/\rho),\nonumber\\
&=\sum_{m=0}^{j} \binom{j}{m} \, \frac{1}{(2\pi \mathrm{i})^{\, m}}\; P^{(m)}(l/\rho) \rho \delta_{l0} \delta_{j-m,0}\nonumber\\
&=\frac{1}{(2\pi \mathrm{i})^{j}} \rho \delta_{l0} P^{(j)}(l/\rho)\nonumber\\
&= \rho \delta_{l0}\delta_{j0},
\end{align}
where by convention, $\binom{j-m}{i} = 0$ whenever $i > j-m$. Consequently, we can derive the following algorithm. 

\begin{algorithm}[H]
\caption*{\textbf{Prediction Algorithm}}
\begin{algorithmic}
\State \textbf{Input:}
\begin{itemize}
    \item A $(r - 1)$-regular generator $\varphi$ for the space $V(\varphi)$, with support $[0, \mu]$.
    \item An offset vector $x=(x_0, x_1, \dots,x_{L-1})$, where $x_0, x_1, \dots, x_{L-1} \in [s, s+1)$ for some $s \in \{0, 1, \dots, \rho - 1\}$ such that $Lr =\rho \geq \mu$.
    \item Choose $\epsilon_p$ for $p=0,1, \dots, \rho-1$ such that $\rho \leq \epsilon_0 < \epsilon_1 < \dots < \epsilon_{\rho-1}.$
    \item Function $f$ with $\tau_r(f^{(i)},W^{-1})_p=\mathcal{O}(W^{-\alpha})$ to be approximated and a scaling parameter $W \geq 1$.
\end{itemize}

\State \textbf{Step 0:} Construct the matrix $C$  given by their entries
\[
C_{ij} = \varphi^{(i \bmod r)}\left(x_{\left\lfloor \frac{i}{r} \right\rfloor} - s - 1 + \rho - j\right), \text{ for } 0 \leq i,j \leq \rho-1. \]

\State \textbf{Step 1:} Verify that $\det C \neq 0$.

\State \textbf{Step 2:} Construct the polyphase matrix $\bm{\Psi} (t)$  with entries defined  for each $0 \leq i,j \leq \rho-1$  by 
$$
[\bm\Psi(t)]_{ij}=
\begin{cases}
\varphi^{(i \bmod r)}\left(x_{\left\lfloor \frac{i}{r} \right\rfloor} - j \right), & \text{if } 0 \leq j \leq s \\
e^{2 \pi it}\varphi^{(i \bmod r)}\left(x_{\left\lfloor \frac{i}{r} \right\rfloor}  + \rho - j\right), & \text{if } s+1 \leq j\leq \rho-1.
\end{cases}
$$
\State \textbf{Step 3:} Compute the inverse of the polyphase matrix $\bm{\Psi}(t).$
\State \textbf{Step 4:} Compute the matrices  $A$ and $B$ whose entries are given by
\[
A_{ij}=\widehat{(\Psi^{-1})}_{ij}(0), \quad B_{ij}=\widehat{(\Psi^{-1})}_{ij}(-1), \quad \text{for } 0 \leq i,j \leq \rho-1. 
\]
\State \textbf{Step 5:} Compute the interpolating kernels $\Theta_{ni}$ given by
\[
\Theta_{ni}(t)=\sum\limits_{q=0}^{s}A_{q, nr+i}\varphi(t-q)+\sum\limits_{q=s+1}^{\rho-1}B_{q,nr+i}\varphi(t+\rho-q),
\]
for $0 \leq n \leq L - 1$ and $0 \leq i \leq r- 1$.
\end{algorithmic}
\end{algorithm}
\begin{algorithm}
\begin{algorithmic}
\State \textbf{Step 6:} Decide  the largest value of $\kappa$ such that
for $j=0,1,\dots,\kappa-1,$
\begin{equation*}
\sum\limits_{i=0}^{r-1}\tbinom{j}{i} i!\sum\limits_{l\in\mathbb{Z}}\sum\limits_{n=0}^{L-1}\left(x_n+\rho l-t\right)^{j-i}\Theta_{ni}(t-\rho l)=\delta_{j0}.
\end{equation*}
\State \textbf{Step 7:} Compute 
$a_p = \prod\limits_{\substack{q=0 \\ q\neq p}}^{\rho-1} \frac{\epsilon_q}{\epsilon_q - \epsilon_p}, \text{ for }  p=0,1,\dots,\rho-1.$
\State \textbf{Step 8:} Compute the modified interpolating kernels $ \widetilde{\Theta}_{ni}(t),$ given by
\begin{equation*}
  \widetilde{\Theta}_{ni}(t)
=\sum_{p=0}^{\rho-1} a_p\, \Theta_{ni}(t - \epsilon_p), \text{ for } 0\leq n\leq L-1,~0\leq i\leq r-1.
\end{equation*}
\State \textbf{Output:}
\begin{itemize}
    \item The approximation of $f$ is given by \eqref{prediction operator}
with the error estimate
\begin{equation*}
\|\widetilde{S}_Wf-f\|_p= \mathcal{O}(W^{-\alpha})\quad (W\to\infty).
\end{equation*}
    \item The number of sample points from the past needed to evaluate $\widetilde{S_W} f$ at any given time is at most 
    $$\rho\left(2 + \left\lfloor \frac{\mu - 1+\epsilon_{p-1}-\epsilon_0}{\rho} \right\rfloor \right).$$
\end{itemize}

\end{algorithmic}
\end{algorithm}

\subsection{Numerical Implementation and Simulation}
We present the validation of theoretical results, focusing on cubic splines and the Daubechies scaling function of order 3. These functions are selected due to their fundamental role in approximation theory and their practical relevance in various applications such as signal processing and numerical analysis.

Simulations demonstrate the effectiveness of the proposed prediction algorithm when applied to these generators. To facilitate further exploration and deepen the understanding of the underlying theory, we provide MATLAB code available at the following link: \href{https://github.com/sreyasoman201/SreyaRiyaAntony}{MATLAB code for prediction algorithm}. This code can easily be adapted to different scenarios by modifying the generator and adjusting key parameters, allowing users to test additional examples and expand the results to a broader range of signals.

For the sake of simplicity, we assume that $s=0$.
We have chosen the offset vector $x=(x_0, x_1, \dots, x_{L-1})$ in two ways: equally spaced nodes or Chebyshev nodes. Chebyshev nodes are chosen because they are the best choice of interpolation points to minimize the error estimate of the Lagrange interpolation.

Recall that the Chebyshev polynomials of the first kind $T_k(x),$ $k\geq0$ on the interval $[-1, 1]$ are defined by the formula
$$T_k(x)=\cos(k\cos^{-1}x).$$
The polynomial $T_k(x)$ is of degree $k$ and all its roots lie in the interval $[-1, 1]$ which are given by $$x_n=\cos\frac{(2n+1)\pi}{2k},~n=0,1,\dots,k-1.$$
For an arbitrary interval $[a, b],$ the points 
$$x_n=\frac{a+b}{2}-\frac{b-a}{2}\cos\frac{(2n+1)\pi}{2k},~n=0,1,\dots,k-1,$$
are the best choice of interpolation points to minimize the error estimate of the Lagrange interpolation on this interval. The points $x_{k}$'s are called the Chebyshev nodes in the interval $[a, b].$
\newline
\textbf{\underline{Example 1}:}\label{example1}
Let us choose the PNS sets 
\begin{flalign*}
  &&X&=4\mathbb{Z}\cup \left(\frac{1}{4}+4\mathbb{Z}\right)\cup\left(\frac{1}{2}+4\mathbb{Z}\right)\cup \left(\frac{3}{4}+4\mathbb{Z}\right)\\
  \text{and} && Y&=\bigcup\limits_{n=0}^3 (x_n+4\mathbb{Z}),\text{ where } x_n=\frac{1}{2}-\frac{1}{2}\cos\frac{(2n+1)\pi}{8},~n=0,1,2,3.
\end{flalign*}
In this case, $L=4$ and $\rho=4$. 
 When $r=1$, it follows from Corollary \ref{cor 3.1} that both $X$ and $Y$ are complete interpolation sets of order $0$ for the space $V(Q_4)$. For the sampling set $X$, the polyphase matrix and its inverse are respectively given by
\begin{equation*}
\renewcommand{\arraystretch}{1.4}
    \bm{\Psi}(x) = 
\begin{bmatrix}
0 & \frac{z}{6} & \frac{2z}{3} & \frac{z}{6} \\ 
\frac{1}{384} & \frac{9z}{128} & \frac{235z}{384} & \frac{121z}{384}\\ 
\frac{1}{48} & \frac{z}{48} & \frac{23z}{48} & \frac{23z}{48} \\ 
\frac{9}{128} & \frac{z}{384} & \frac{121z}{384} & \frac{235z}{384} \\
\end{bmatrix} \text{ and }
\bm{\Psi}^{-1}(x) = 
\begin{bmatrix}
-19 & \frac{208}{3} & \frac{-260}{3} & \frac{112}{3} \\ 
\frac{19}{z} & -\frac{116}{3z} & \frac{82}{3z} & -\frac{20}{3z}  \\
-\frac{13}{3z} & \frac{40}{3z} & -\frac{32}{3z} & \frac{8}{3z} \\
\frac{13}{3z} & -\frac{44}{3z} & \frac{46}{3z} & -\frac{4}{z}
\end{bmatrix}.
\end{equation*}
The Fourier coefficients of $\bm{\Psi}^{-1}$ are given by
\begin{equation*}
\renewcommand{\arraystretch}{1.4}
\widehat{\bm{\Psi}}^{-1}(v) =
\begin{bmatrix}
-19 \delta_{v0} & \frac{208\delta_{v0}}{3} & \frac{-260 \delta_{v0}}{3}&\frac{112\delta_{v0}}{3}\\ 
19\delta_{v,-1} & -\frac{116 \delta_{v.-1}}{3} & \frac{82\delta_{v,-1}}{3} & \frac{-20\delta_{v,-1}}{3}\\ 
\frac{-13\delta_{v,-1}}{3} & \frac{40 \delta_{v.-1}}{3} & -\frac{32\delta_{v,-1}}{3} & \frac{8\delta_{v,-1}}{3}\\ 
\frac{13\delta_{v,-1}}{3} & \frac{-
44 \delta_{v.-1}}{3} & \frac{46\delta_{v,-1}}{3} & -4\delta_{v,-1}
\end{bmatrix},~v\in\mathbb{Z}.
\end{equation*}
Hence every $f\in V(Q_4)$ can be written as
\begin{align*}
f(t)=\sum_{n=0}^{3}\sum_{l\in\mathbb{Z}}f(x_n+4l)\Theta_{n0}(t-4l),
\end{align*}
where 
\begin{align*}
\Theta_{00}(t) &= -19 Q_4(t) + 19 Q_4(t+3) - \frac{13}{3} Q_4(t+2) + \frac{13}{3} Q_4(t+1), \\
\Theta_{10}(t) &= \frac{208}{3} Q_4(t) - \frac{116}{3} Q_4(t+3) + \frac{40}{3} Q_4(t+2) - \frac{44}{3} Q_4(t+1), \\
\Theta_{20}(t) &= -\frac{260}{3} Q_4(t) + \frac{82}{3} Q_4(t+3) - \frac{32}{3} Q_4(t+2) + \frac{46}{3} Q_4(t+1), \\
\Theta_{30}(t) &= \frac{112}{3} Q_4(t) - \frac{20}{3} Q_4(t+3) + \frac{8}{3} Q_4(t+2) - 4 Q_4(t+1),
\end{align*}
are the interpolating kernels with support $[-3,4]$. 

Similarly, the polyphase matrix and its inverse corresponding to the sampling set $Y$ can be computed.  Due to numerical precision, the matrix entries contain more digits and are therefore not explicitly included in the paper. However, both the matrices and the corresponding interpolating kernels can be generated using our MATLAB code. The graphs of the interpolating kernels for the sample sets $X$ and $Y$ are shown in Figure \ref{fig:interpolating_kernels}.
\begin{figure}[H]
\centering
\begin{subfigure}[!b]{0.5\textwidth}
    \centering
    \includegraphics[width=\textwidth]{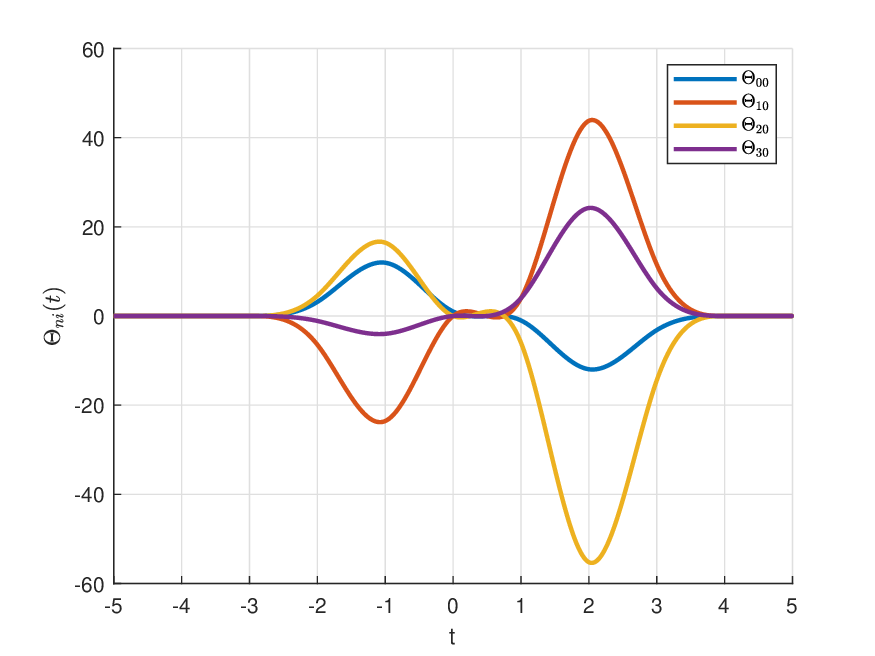}
    \caption{Equally spaced nodes}
\end{subfigure}
\hspace{-1cm}
\begin{subfigure}[!b]{0.5\textwidth}
    \centering
    \includegraphics[width=\textwidth]{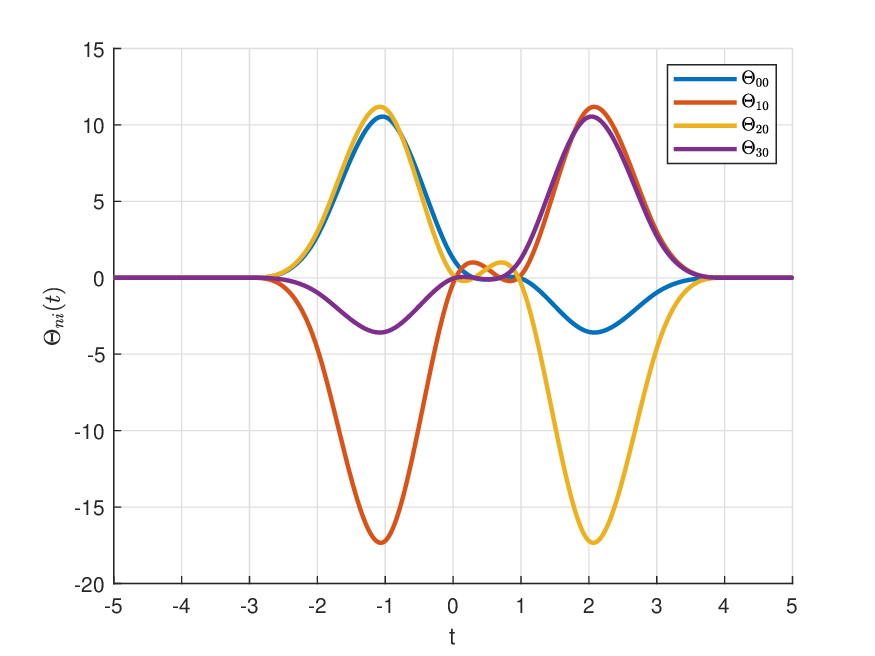}
    \caption{Chebyshev nodes}
\end{subfigure}
\caption{Interpolating kernels for Example 1}
\label{fig:interpolating_kernels}
\end{figure}
We now choose
$$
\epsilon_p = 4 + 0.25p, \qquad p = 0,1,2,3.
$$
Using~\eqref{a_p2}, we obtain
$$
(a_0, a_1, a_2, a_3) = (969,\,-2736,\,2584,\,-816).
$$
Thus, for \(0 \leq n \leq 3\), the modified interpolating kernels $\widetilde{\Theta}_{n0}$ are given by
$$\widetilde{\Theta}_{n0}(t) = \sum_{p=0}^{3} a_p\, \Theta_{n0}(t - \epsilon_p),$$
with support $[1,8.75].$
We also verified that $S_W$ satisfies the reproducing polynomial property of order 3 and hence $\widetilde{S_W}$ satisfies it as well. This property can also be checked using our MATLAB code for the sampling operator by testing it on monomials. Hence, it follows  that if  $f\in\Lambda_4^p$ such that  $\tau_4(f,W^{-1})_p=\mathcal{O}(W^{-\alpha})$ as $W\to \infty$,  then
\begin{equation*}
\|\widetilde{S}_Wf-f\|_p= \mathcal{O}(W^{-\alpha}) \quad (W\to\infty).\\
\end{equation*}
\noindent
\textbf{\underline{Example 2}:}
Consider the PNS set $X= (\frac{1}{2}+4\mathbb{Z})\cup (\frac{3}{4}+4\mathbb{Z})$ for the space $V(Q_4)$.  In this case, $L=2$ and $\rho=4$. When $r=2$, it follows from Corollary \ref{cor 3.1} that $X$ is CIS of order $1$ for $V(Q_4)$. The polyphase matrix and its inverse for the sampling set are respectively given by
\begin{equation*}
\renewcommand{\arraystretch}{1.4}
\bm{\Psi}(x) = 
\begin{bmatrix}
\frac{1}{48} & \frac{z}{48} & \frac{23z}{48} & \frac{23z}{48} \\ 
\frac{1}{8} & -\frac{z}{8} & -\frac{5z}{8} & \frac{5z}{8} \\ 
\frac{9}{128} & \frac{z}{384} & \frac{121z}{384} & \frac{235z}{384} \\ 
\frac{9}{32} & -\frac{z}{32} & -\frac{21z}{32} & \frac{13z}{32}
\end{bmatrix} \text{ and }
\bm{\Psi}^{-1}(x) = 
\begin{bmatrix}
149 & \frac{97}{6} & -148 & \frac{67}{3} \\ 
-\frac{331}{z} & -\frac{281}{6z} & \frac{332}{z} & -\frac{113}{3z} \\ 
\frac{53}{z} & \frac{37}{6z} & -\frac{52}{z} & \frac{19}{3z} \\ 
-\frac{43}{z} & -\frac{29}{6z} & \frac{44}{z} & -\frac{17}{3z}
\end{bmatrix}.
\end{equation*}
The Fourier coefficients of $\bm{\Psi}^{-1}$ are given by
\begin{equation*}
    \renewcommand{\arraystretch}{1.4}
\widehat{\bm{\Psi}}^{-1}(v) =
\begin{bmatrix}
149 \delta_{v0} & \frac{97\delta_{v0}}{6} & -148 \delta_{v0}&\frac{67\delta_{v0}}{3}\\ 
-331\delta_{v,-1} & -\frac{281 \delta_{v.-1}}{6} & 332\delta_{v,-1} & \frac{-113\delta_{v,-1}}{3}\\ 
53\delta_{v,-1} & \frac{37 \delta_{v.-1}}{6} & -52\delta_{v,-1} & \frac{19\delta_{v,-1}}{3}\\ 
-43\delta_{v,-1} & \frac{-
29 \delta_{v.-1}}{6} & 44\delta_{v,-1} & \frac{-17\delta_{v,-1}}{3}
\end{bmatrix},~v\in\mathbb{Z}.
\end{equation*}
Hence
\begin{align*}
\Theta_{00}(t) &= 149 Q_4(t) - 331 Q_4(t+3) + 53 Q_4(t+2) - 43 Q_4(t+1), \\
\Theta_{01}(t) &= \frac{97}{6} Q_4(t) - \frac{281}{6} Q_4(t+3) + \frac{37}{6} Q_4(t+2) - \frac{29}{6} Q_4(t+1), \\
\Theta_{10}(t) &= -148 Q_4(t) + 332 Q_4(t+3) - 52 Q_4(t+2) + 44 Q_4(t+1), \\
\Theta_{11}(t) &= \frac{67}{3} Q_4(t) - \frac{113}{3} Q_4(t+3) + \frac{19}{3} Q_4(t+2) - \frac{17}{3} Q_4(t+1).
\end{align*}
are the interpolating kernels with support $[-3,4]$.
Using the same shifts $\epsilon_p$ and the coefficients $a_p$ as in the previous example, the modified interpolating kernels $\widetilde{\Theta}_{ni}$, $n,i=0,1$ are given by
$$\widetilde{\Theta}_{ni}(t) = \sum_{p=0}^{3} a_p\, \Theta_{ni}(t - \epsilon_p),$$
with support $[1,8.75].$
Again using our MATLAB code, we verified that  $S_W$ satisfies the reproducing polynomial property of order $3$.
Applying Theorem \ref{approximation  operator}, it follows  that if  $f\in\Lambda_4^p$ such that  $\tau_4(f^{(i)},W^{-1})_p=\mathcal{O}(W^{-\alpha})$ as $W\to \infty$ for each $i=0,1,$ then
\begin{equation*}
\|\widetilde{S}_Wf-f\|_p= \mathcal{O}(W^{-\alpha}) \quad (W\to\infty).\\
\end{equation*}
\noindent
\textbf{\underline{Example 3}:}
Consider the PNS set $$X=\bigcup\limits_{n=0}^4 (x_n+5\mathbb{Z}),\text{ where } x_n=\frac{1}{2}-\frac{1}{2}\cos\frac{(2n+1)\pi}{10},~n=0,1,2,3,4.$$ In this case, $L=5$ and $\rho=5.$ We choose the Daubechies scaling function of order 3 (db3) as the generator. When $r=1$, it follows from Corollary \ref{cor 3.1} that $X$ is a CIS of order 0 for $V(\textrm{db3})$. The polyphase matrix and its inverse corresponding to the sampling set $X$ can be computed using our MATLAB code. The graphs of interpolating kernels are given in Fig. \ref{fig:interpolating_kernels2b}.
We now choose
$$
\epsilon_p = 5 +5p, \qquad p = 0,1,2,3,4.
$$
Using~\eqref{a_p2}, we obtain
$$
(a_0, a_1, a_2, a_3, a_4) = (5, -10, 10, -5, 1).
$$
Thus, for $0 \leq n \leq 4,$  the modified interpolating kernels $\widetilde{\Theta}_{n0}$ are given by
$$\widetilde{\Theta}_{n0}(t) = \sum_{p=0}^{4} a_p\, \Theta_{n0}(t - \epsilon_p),$$
with  support  $[1,11].$
We also verified that $S_W $ satisfies the reproducing polynomial property of order 3. Applying Theorem \ref{approximation  operator}, it follows  that if  $f\in\Lambda_4^p$ such that  $\tau_4(f,W^{-1})_p=\mathcal{O}(W^{-\alpha})$ as $W\to\infty$, then
\begin{equation*}
\|\widetilde{S}_Wf-f\|_p= \mathcal{O}(W^{-\alpha}) \quad (W\to\infty).\\
\end{equation*}
\noindent
\begin{figure}[H]
\centering
\begin{subfigure}[b]{0.48\textwidth}
    \centering
    \includegraphics[width=\textwidth]{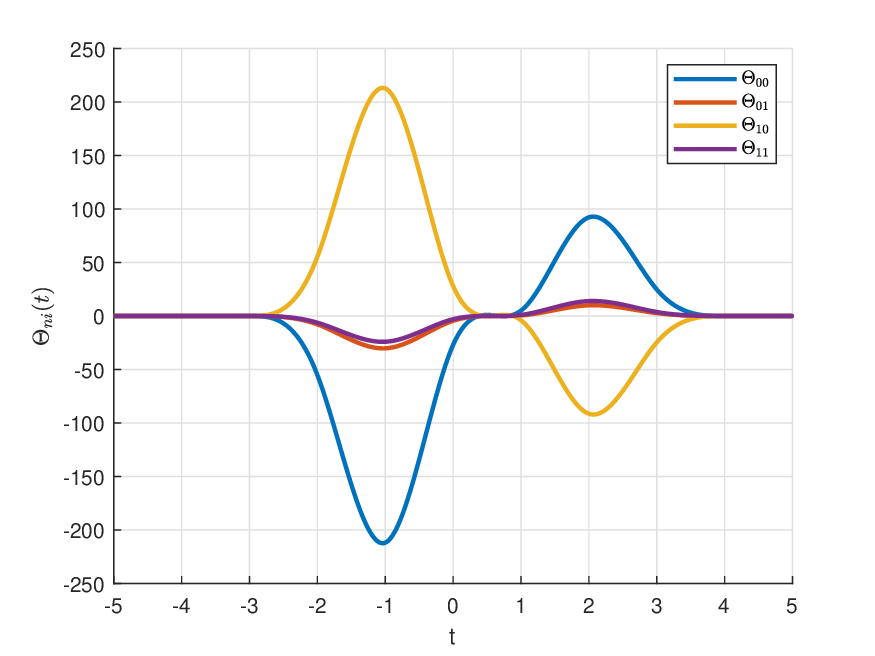}
    \caption{Example 2}
    \label{fig:interpolating_kernels2a}
\end{subfigure}
\hfill
\begin{subfigure}[b]{0.48\textwidth}
    \centering
    \includegraphics[width=\textwidth]{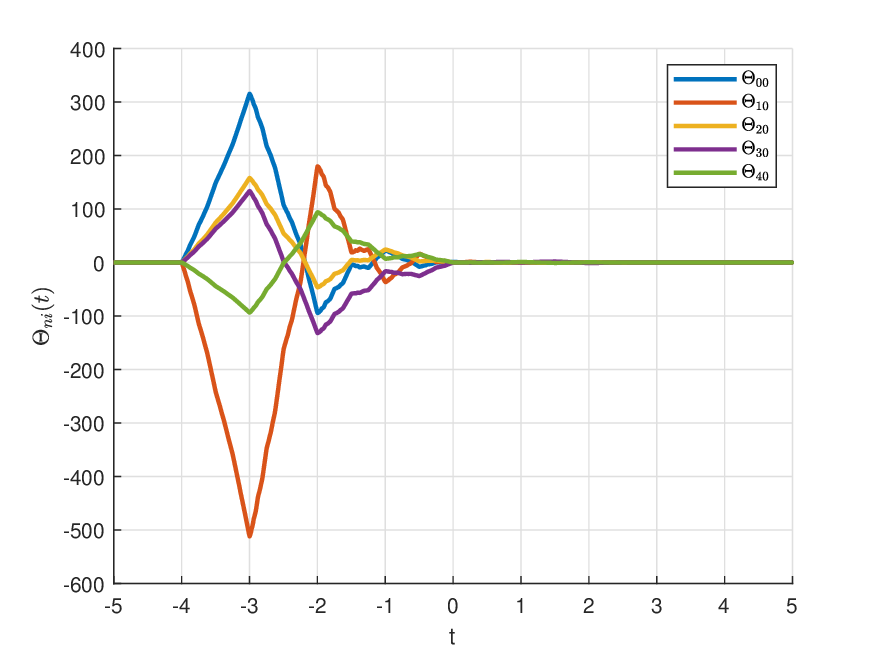}
    \caption{Example 3}
    \label{fig:interpolating_kernels2b}
\end{subfigure}
\caption{Interpolating kernels}
\label{fig:interpolating_kernels2}
\end{figure}

We now approximate two real-valued functions $f$ and $g$ using our prediction operator $\widetilde{S}_W$.
The first function $f :\mathbb{R}\to \mathbb{R}$ is given by
$$f(t)=e^{-t^2/4}\sin(2\pi t).$$
It is an infinitely differentiable function and $f_1\in\Lambda_{r}^p$ for any $r\in\mathbb{N}$ for any of the sample sets considered in the examples.
The second function is a discontinuous function given by
$$
g(x)=
\begin{cases}
-\dfrac{1}{2}x^{3}+2, & \text{if } -1.5<x<3,\\
0, & \text{otherwise}.
\end{cases}
$$
It belongs to $\Lambda_{1}^p$ for any of the sample sets considered in the examples. 

Figure \ref{figure3} displays the approximation of $f$ using  $\widetilde{S}_{W}f$ for various values of $W$, based on the PNS set with an offset vector consisting of equally spaced nodes, as described in Example 1. Similarly, Figure \ref{figure4} shows the approximation of $g$ under the same conditions. It is well known that approximating a discontinuous signal with a continuous one leads to poor behavior near points of discontinuity. This issue becomes particularly evident when using series approximations based on the interpolating kernels. The oscillations that appear near jump discontinuities are known as the Gibbs phenomenon. Increasing $W$ only narrows the region over which these oscillations occur; it does not reduce the height of the overshoot. This behavior has been observed in prior studies \cite{BBSV2}. Since our analysis is based on mean-square accuracy, the Gibbs phenomenon does not pose a significant problem in this context.

Figure \ref{figure5}  presents the approximations of $f$ using $ \widetilde{S}_{W}$ for different values of $W$, but this time based on the PNS set with an offset vector consisting of Chebyshev nodes, also from Example 1.
Table \ref{Approximation error} compares the $L^2$-approximation errors for $f$  based on both types of PNS sets.
Figure \ref{figure6} shows the approximation of $f$ using $\widetilde{S}_{W}f$ for various  $W$, based on the PNS set with equally spaced nodes, as described in Example 2.
Figures \ref{figure7}  illustrate the approximations of $f$  for varying $W$, based on the PNS set with Chebyshev nodes, as described in Example 3.
\begin{figure}[H]
    \centering
    \begin{subfigure}[b]{0.45\textwidth}
        \centering
        \includegraphics[height=4.5cm,width=\textwidth]{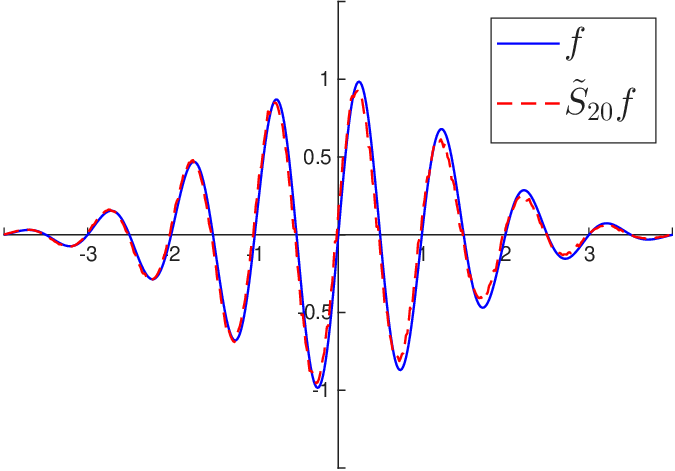}
    \end{subfigure}
    \hspace{0.5cm}
    \begin{subfigure}[b]{0.45\textwidth}
        \centering
        \includegraphics[height=4.5cm, width=\textwidth]{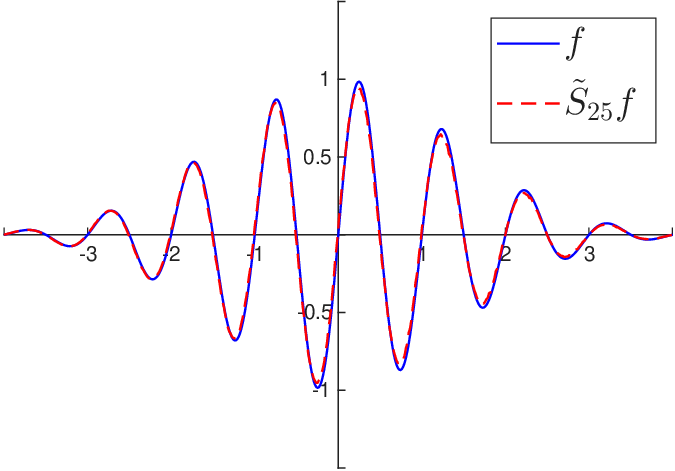}
    \end{subfigure}
    \begin{subfigure}[b]{0.45\textwidth}
        \centering
\includegraphics[height=4.5cm,width=\textwidth]{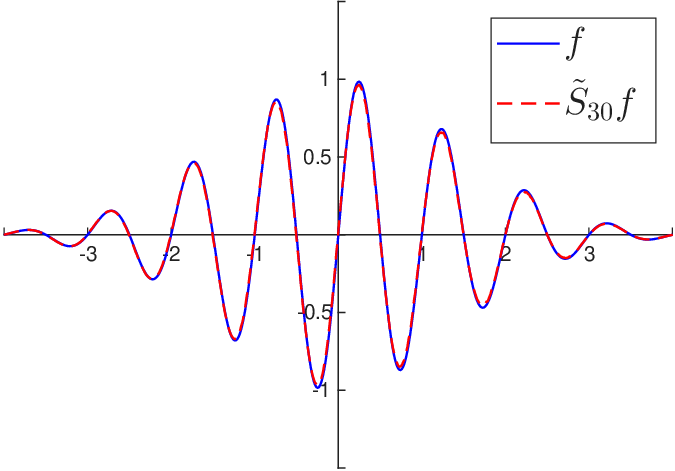}
    \end{subfigure}

    \caption{Approximations of $f$ using $\widetilde{S}_{W}f$ for equally spaced nodes in Example 1.}
    \label{figure3}
\end{figure}
\begin{figure}[H]
    \centering
    \begin{subfigure}[!b]{0.5\textwidth}
        \centering
        \includegraphics[height=4.5cm, width=\textwidth]{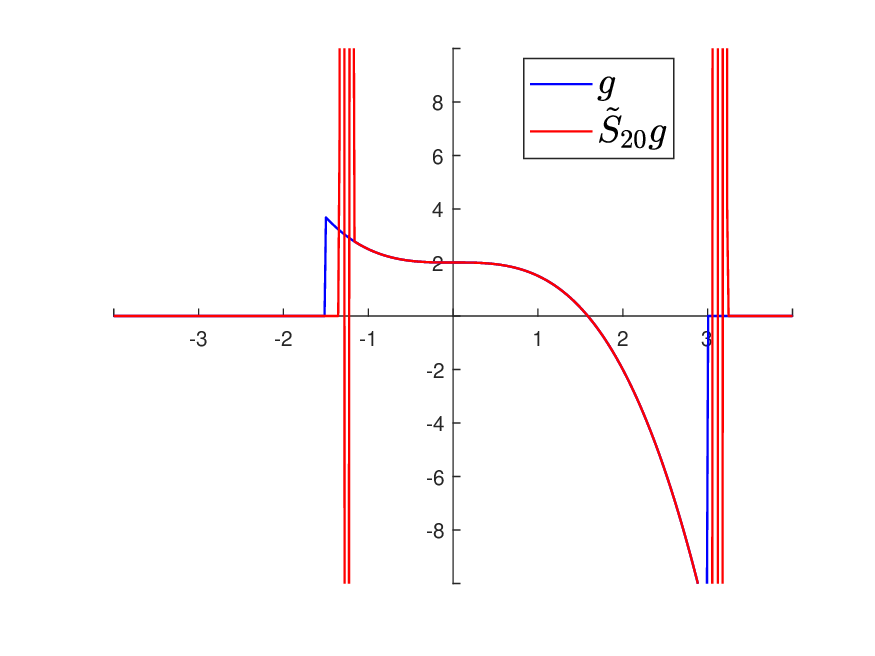}
    \end{subfigure}
    \hspace{-0.8cm}
    \begin{subfigure}[!b]{0.5\textwidth}
        \centering
        \includegraphics[height=4.5cm, width=\textwidth]{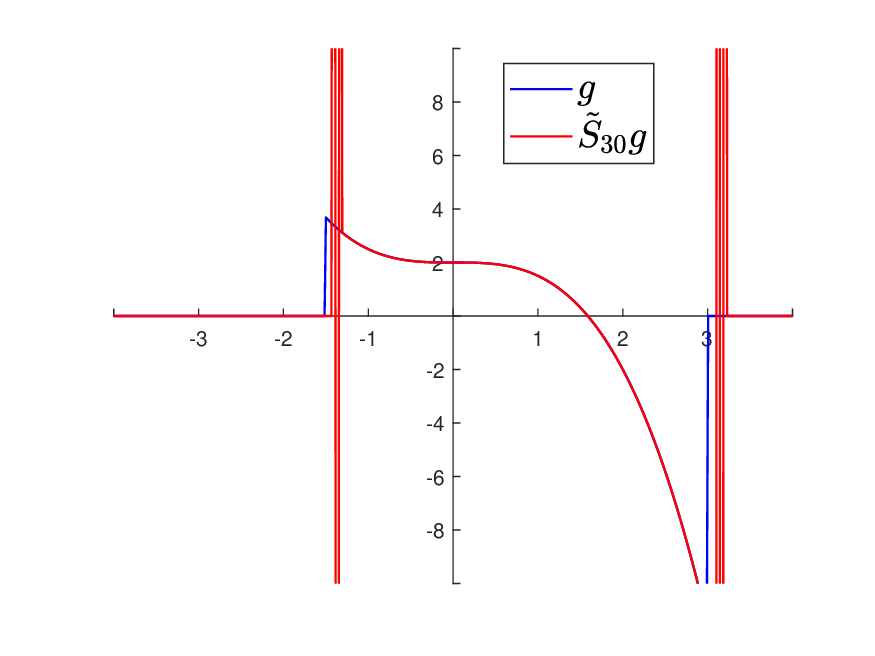}
    \end{subfigure}
    
    \caption{Approximations of $g$ based on $\widetilde{S}_{W}g$ for equally spaced nodes in Example 1.}\label{figure4}
\end{figure}

\begin{figure}[H]
    \centering

    \begin{subfigure}[b]{0.45\textwidth}
        \centering
        \includegraphics[height=4.5cm,width=\textwidth]{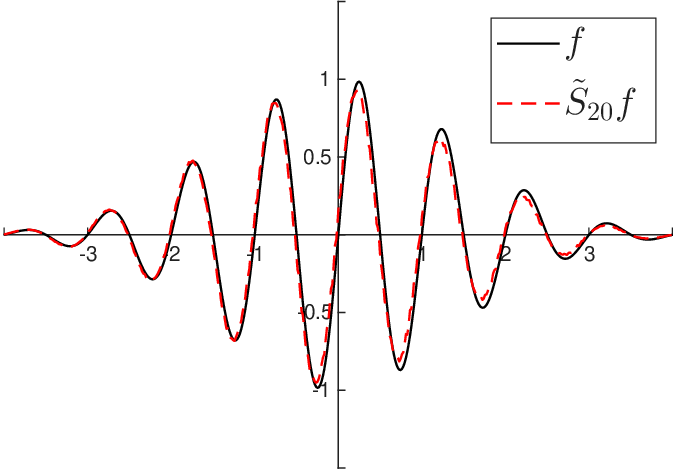}
    \end{subfigure}
    \hspace{0.5cm}
    \begin{subfigure}[b]{0.45\textwidth}
        \centering
        \includegraphics[height=4.5cm,width=\textwidth]{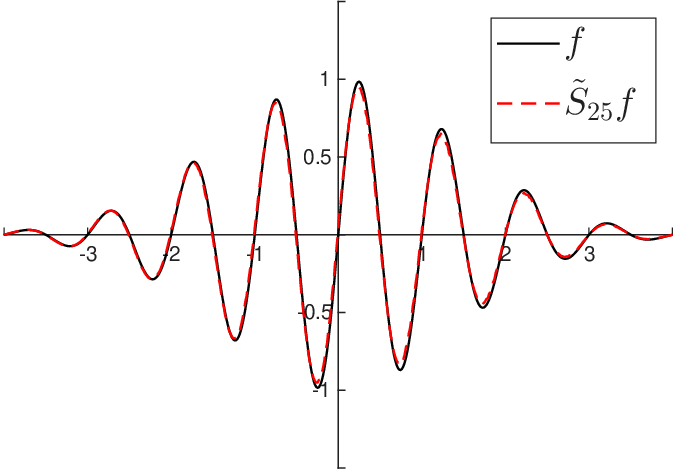}
    \end{subfigure}
    \begin{subfigure}[b]{0.45\textwidth}
        \centering
        \includegraphics[height=4.5cm,width=\textwidth]{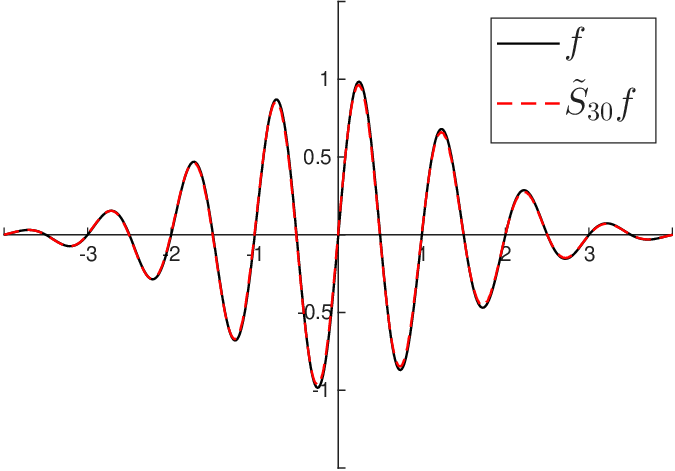}
    \end{subfigure}

    \caption{Approximations of $f$ using $\widetilde{S}_{W}f$ for Chebyshev nodes in Example 1.}
    \label{figure5}
\end{figure}

\begin{figure}[H]
    \centering

    \begin{subfigure}[b]{0.45\textwidth}
        \centering
        \includegraphics[height=4.5cm,width=\textwidth]{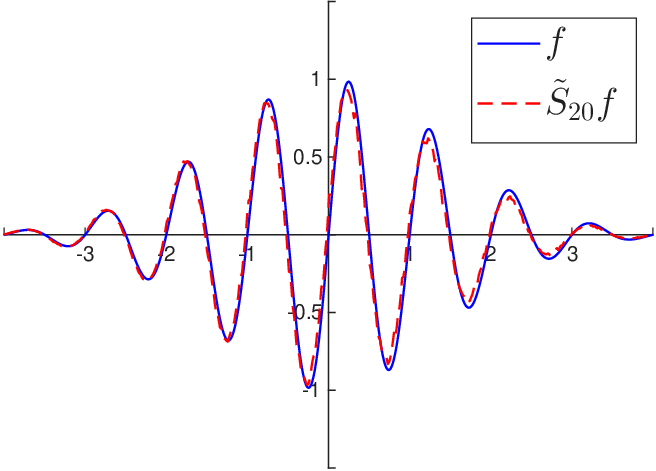}
    \end{subfigure}
    \hspace{0.5cm}
    \begin{subfigure}[b]{0.45\textwidth}
        \centering
        \includegraphics[height=4.5cm,width=\textwidth]{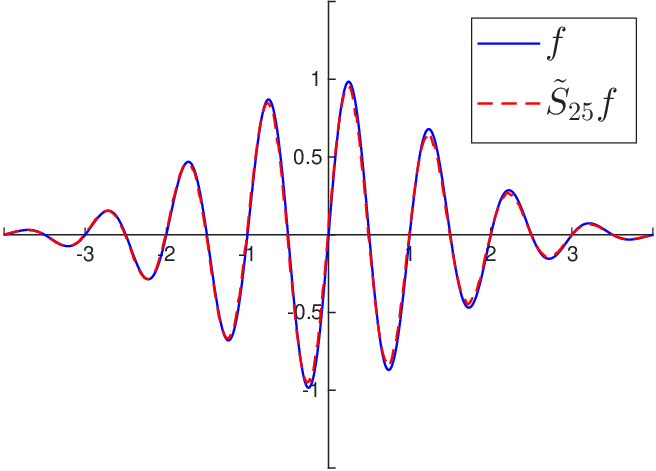}
    \end{subfigure}
\end{figure}
\begin{figure}\ContinuedFloat
    \begin{subfigure}[b]{0.45\textwidth}
        \centering
        \includegraphics[height=4.5cm,width=\textwidth]{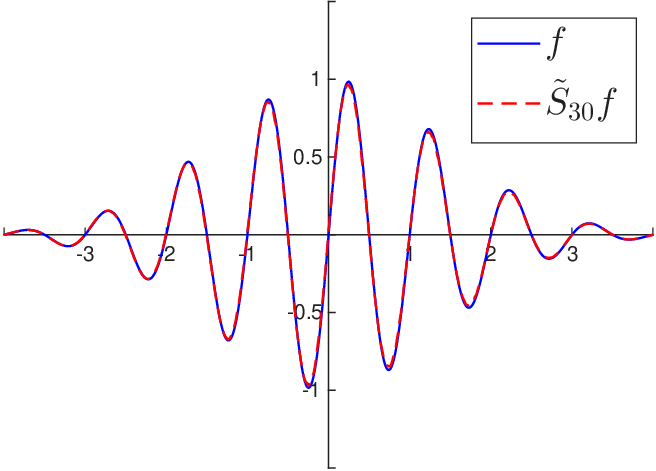}
    \end{subfigure}
    \caption{Approximations of $f$ based on $\widetilde{S}_{W}f$ for equally spaced nodes in Example~2.}
    \label{figure6}
\end{figure}

\begin{table}[h]
\centering
\begin{tabular}{|c||c|c|}
\hline
  $W$ 
  & \multicolumn{2}{c||}{$\|\widetilde{S}_Wf - f\|_2$} \\
  \hline
& Equally spaced & Chebyshev  \\
  \hline
   5 &  32.8862 & 32.1176 \\
  \hline
    7 &  9.9272 & 9.8323 \\
   \hline
   10 &  2.6441 & 2.6197 \\
   \hline
   15 &  0.55445 & 0.54868 \\
  \hline
  20 & 0.17917  & 0.17726 \\
  \hline
  25 & 0.073676 & 0.073303 \\
  \hline
  30 & 0.035946 & 0.03555 \\
  \hline
\end{tabular}

\caption{$L^2$-approximation error for $f$ in Example 1.}\label{Approximation error}
\end{table}

\begin{figure}[H]
    \centering
    \begin{subfigure}[!b]{0.5\textwidth}
        \centering
        \includegraphics[height=4.5cm,width=\textwidth]{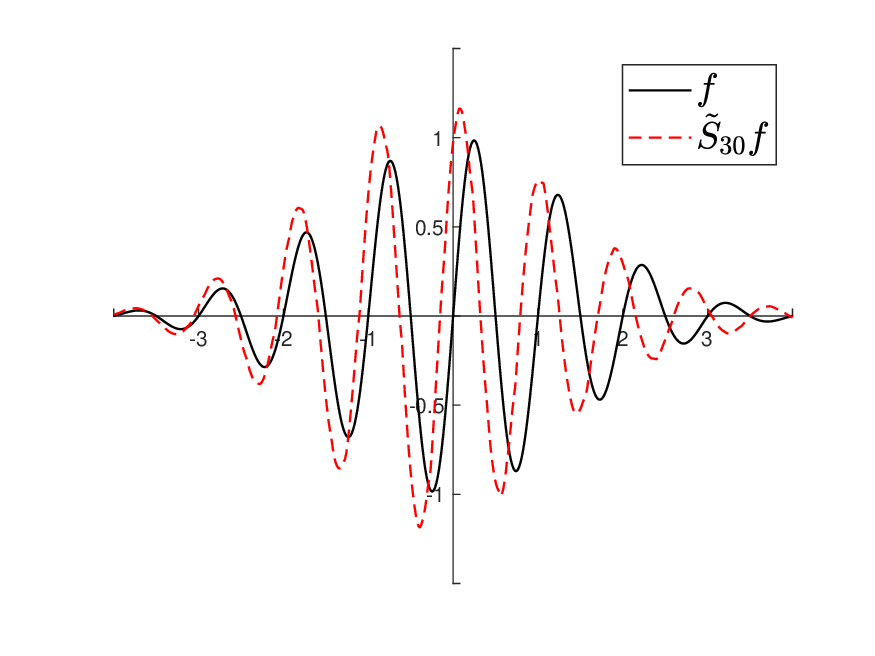}
    \end{subfigure}
    \hspace{-1cm}
    \begin{subfigure}[!b]{0.5\textwidth}
        \centering
        \includegraphics[height=4.5cm,width=\textwidth]{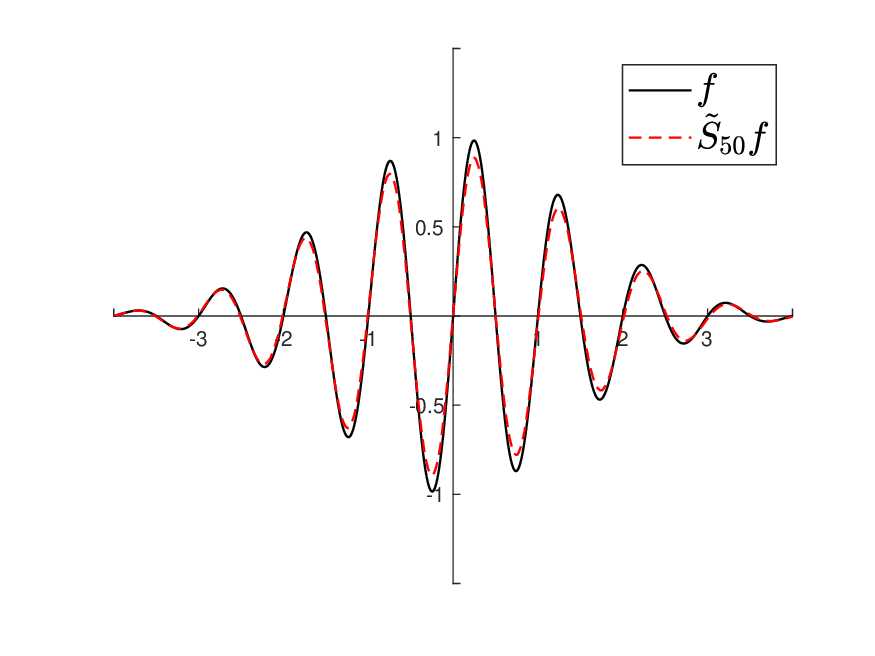}
    \end{subfigure}
    \vspace{-1cm}
    \begin{subfigure}[b]{0.5\textwidth}
        \centering
        \includegraphics[height=4.5cm,width=\textwidth]{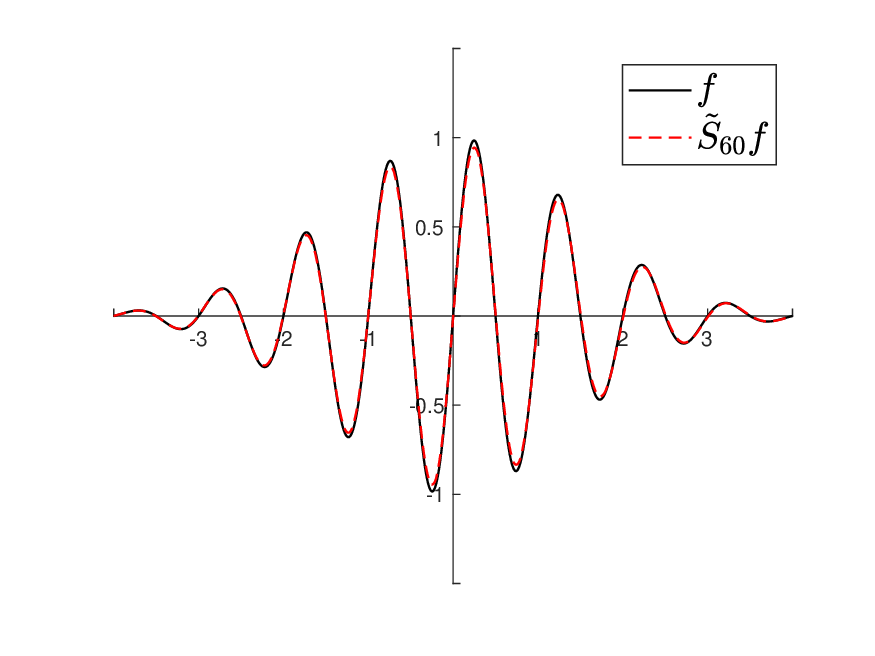}
    \end{subfigure}

    \caption{Approximations of $f$ based on $\widetilde{S}_{W}f$ for Chebyshev spaced nodes in Example 3.}\label{figure7}
\end{figure}

\section{Appendix}

\underline{\textbf{Proof of Lemma \ref{lemma4.1}:}}
By a change of variable, it is enough to prove the result for $W=1.$ 
Let us define
$$\displaystyle f_j(t):=\sum\limits_{i=0}^{r-1}\tbinom{j}{i}i!\sum\limits_{l\in\mathbb{Z}}\sum\limits_{n=0}^{L-1}\left(x_n+\rho l-t\right)^{j-i}\Theta_{ni}(t-\rho l).$$\\
Applying the binomial theorem, we get
\begin{align*}
f_j(t)
&=\sum\limits_{l\in\mathbb{Z}}\sum\limits_{i=0}^{r-1}\tbinom{j}{i} i!\sum\limits_{n=0}^{L-1}\sum\limits_{k=0}^{j-i}\tbinom{j-i}{k}\left(x_n+\rho l\right)^{k}(-t)^{j-i-k}\Theta_{ni}(t-\rho l)\\
&=\sum\limits_{l\in\mathbb{Z}}\sum\limits_{i=0}^{r-1}\tbinom{j}{i}i!\sum\limits_{s=i}^{j}\tbinom{j-i}{s-i}(-t)^{j-s}\sum\limits_{n=0}^{L-1}\left(x_n+\rho l\right)^{s-i}\Theta_{ni}(t-\rho l)\\
&=\sum\limits_{l\in\mathbb{Z}}\sum\limits_{i=0}^{r-1}i!\sum\limits_{s=i}^{j}\tbinom{j-i}{s-i}\tbinom{j}{i}(-t)^{j-s}\sum\limits_{n=0}^{L-1}\left(x_n+\rho l\right)^{s-i}\Theta_{ni}(t-\rho l)\\
&=\sum\limits_{l\in\mathbb{Z}}\sum\limits_{i=0}^{r-1}i!\sum\limits_{s=0}^{j}\tbinom{j}{s}\tbinom{s}{i}(-t)^{j-s}\sum\limits_{n=0}^{L-1}\left(x_n+\rho l\right)^{s-i}\Theta_{ni}(t-\rho l)\\
&=\sum\limits_{s=0}^{j}\tbinom{j}{s}(-t)^{j-s}\sum\limits_{l\in\mathbb{Z}}\sum\limits_{i=0}^{r-1}\tbinom{s}{i}i!\sum\limits_{n=0}^{L-1}\left(x_n+\rho l\right)^{s-i}\Theta_{ni}(t-\rho l).
\end{align*}
It is now straightforward to conclude that statements $(i)$ and $(ii)$ are equivalent by selecting $f(t)=t^s$.
Notice that $f_j(t)$ is a periodic function with period $\rho$. We can easily verify that 
\begin{align*}
\widehat{f_j}(l)&:=\frac{1}{\rho}\int\limits_{0}^{\rho}f_j(t)e^{-2\pi\mathrm{i}lt/\rho}~dt\\
&=\frac{1}{\rho}\sum\limits_{i=0}^{r-1}\tbinom{j}{i}i!\sum\limits_{m=0}^{j-i}\tbinom{j-i}{m}x_n^m(2\pi\mathrm{i})^{m-j+i} \times\widehat{\Theta_{ni}}^{(j-i-m)}\left(\tfrac{l}{\rho}\right)\\
&= \frac{1}{\rho}\sum\limits_{i=0}^{r-1}\tbinom{j}{i}i!\Big(\frac{1}{2\pi \mathrm{i}}\Big)^{j-i}\sum\limits_{m=0}^{j-i}\tbinom{j-i}{m}(2\pi\mathrm{i}x_n)^{m} \times\widehat{\Theta_{ni}}^{(j-i-m)}\left(\tfrac{l}{\rho}\right)\\
&= \frac{1}{\rho}\sum\limits_{i=0}^{r-1}\tbinom{j}{i}i!\Big(\frac{1}{2\pi \mathrm{i}}\Big)^{j-i}\dfrac{d^{j-i}}{dw^{j-i}}\Bigg(\sum\limits_{n=0}^{L-1}
e^{\, 2\pi \mathrm{i} x_n \left(w - \frac{l}{\rho}\right)}
\, \widehat{\Theta}_{n i}(w)\Bigg)\Bigg|_{w=l/\rho}\\
&= \frac{1}{\rho}\sum\limits_{i=0}^{r-1}\tbinom{j}{i}i!\Big(\frac{1}{2\pi \mathrm{i}}\Big)^{j-i}F_{li}^{(j-i)}(l/\rho).
\end{align*}
Now the equivalence of $(ii)$ and $(iii)$ follows from the uniqueness of the Fourier series.\\\\
\underline{\textbf{Proof of Theorem \ref{approximation operator}:}}
Applying H$\Ddot{o}$lder's inequality with $1/p+1/q= 1$ in  \eqref{sampling series}, we get $|(S_Wf)(t)|$
\begin{align*}
&\leq \sum\limits_{i=0}^{r-1}\frac{1}{W^i}\left\{\sum\limits_{l\in\mathbb{Z}}\sum\limits_{n=0}^{L-1}\left|f^{(i)}\left(\frac{x_n+ \rho l}{W}\right)\right|^p\left|\Theta_{ni}(Wt-\rho l)\right|\right\}^{\tfrac{1}{p}}\left\{\sum\limits_{l\in\mathbb{Z}}\sum\limits_{n=0}^{L-1}\left|\Theta_{ni}(Wt-\rho l)\right|\right\}^{\tfrac{1}{q}}\nonumber\\
&\leq \sum\limits_{i=0}^{r-1}\frac{1}{W^i}\left\{\sum\limits_{l\in\mathbb{Z}}\sum\limits_{n=0}^{L-1}\left|f^{(i)}\left(\frac{x_n+ \rho l}{W}\right)\right|^p\left|\Theta_{ni}(Wt-\rho l)\right|\right\}^{\tfrac{1}{p}}\left\{\sum\limits_{n=0}^{L-1}m(\Theta_{ni})\right\}^{\tfrac{1}{q}},
\end{align*}
where
$$m(\Theta_{ni}):=\sup\limits_{u\in[0, \rho]}\sum\limits_{l\in\mathbb{Z}}|\Theta_{ni}(u-\rho l)|
<\infty.$$
Hence when $W\geq1$, $\displaystyle\int\limits_{-\infty}^{\infty}\left|(S_Wf)(t)\right|^p dt$
\begin{align*}
&\leq\sum\limits_{i=0}^{r-1}\frac{1}{W^i}\left\{\sum\limits_{l\in\mathbb{Z}}\sum\limits_{n=0}^{L-1}\left|f^{(i)}\left(\frac{x_n+ \rho l}{W}\right)\right|^p\int\limits_{-\infty}^\infty\left|\Theta_{ni}(Wt-\rho l)\right|dt\right\}\left\{\sum\limits_{n=0}^{L-1}m(\Theta_{ni})\right\}^{\tfrac{p}{q}}\nonumber\\
&\leq\left\{\sum\limits_{n=0}^{L-1}m(\Theta_{ni})\right\}^{\tfrac{p}{q}}\sum\limits_{i=0}^{r-1}\left\{\sum\limits_{l\in\mathbb{Z}}\sum\limits_{n=0}^{L-1}\left|f^{(i)}\left(\frac{x_n+ \rho l}{W}\right)\right|^p\dfrac{1}{W}\|\Theta_{ni}\|_1\right\}.\nonumber
\end{align*}
Hence there exists a constant $K_1>0$ independent of $W$ such that
\begin{equation*}\label{inequalitySWf}
\|S_Wf\|_{p}\leq K_1\|f\|_{\ell^p_{r}(W)}, \text{ for every } f\in \Lambda^p_{r} \text{ and } W\geq 1.
\end{equation*}
Let $g\in  W_p^r.$ Then by Taylor's formula, we have 
\begin{align*}
g^{(i)}\left( \tfrac{x_n + \rho l}{W} \right) &=  
\sum\limits_{\nu = i}^{r-1} \frac{g^{(\nu)}(t)}{(\nu - i)!} 
\left( \frac{x_n + \rho l}{W} - t \right)^{\nu - i} + \tfrac{1}{(r-1-i)!} 
\int\limits_{t}^{\tfrac{x_n + \rho l}{W}} g^{(r)}(u) 
\left( \tfrac{x_n + \rho l}{W} - u \right)^{r-1-i} \, du,
\end{align*} 
for $l \in \mathbb{Z}$ and $t \in \mathbb{R}.$
It follows that $(S_Wg)(t)$
\begin{align*}
&=\sum\limits_{l\in\mathbb{Z}}\Bigg(\sum\limits_{i=0}^{r-1}\sum\limits_{n=0}^{L-1}\frac{1}{W^i}\sum\limits_{\nu=i}^{r-1}\frac{g^{(\nu)}(t)}{(\nu-i)!}\left(\frac{x_n+\rho l}{W}-t\right)^{\nu-i}\Theta_{ni}(Wt- \rho l)\\
&\hspace{0.5cm}+\sum\limits_{i=0}^{r-1}\sum\limits_{n=0}^{L-1}\frac{1}{W^i(r-1-i)!}\sum\limits_{l\in\mathbb{Z}}\Theta_{ni}(Wt-\rho l)\int\limits_{t}^{\frac{x_n+ \rho l}{W}}g^{(r)}(u)\left(\frac{x_n+ \rho l}{W}-u\right)^{r-1-i} ~du\Bigg)\\
&=\sum\limits_{\nu=0}^{r-1}\tfrac{g^{(\nu)}(t)}{\nu!}\sum\limits_{l\in\mathbb{Z}}\Bigg(\sum\limits_{i=0}^{r-1}\sum\limits_{n=0}^{L-1}\tfrac{i!}{W^i}\tbinom{\nu}{i}\left(\tfrac{x_n+ \rho l}{W}-t\right)^{\nu-i}\Theta_{ni}(Wt-\rho l)\\
&\hspace{0.5cm}+\sum\limits_{i=0}^{r-1}\sum\limits_{n=0}^{L-1}\frac{1}{W^i(r-1-i)!}\sum\limits_{l\in\mathbb{Z}}\Theta_{ni}(Wt-\rho l)\int\limits_{t}^{\frac{x_n+\rho l}{W}}g^{(r)}(u)\left(\frac{x_n+\rho l}{W}-u\right)^{r-1-i} ~du\Bigg).
\end{align*}
By applying  Lemma \ref{lemma4.1}, we get $(S_Wg)(t)-g(t)$
\begin{align*}
&= \sum\limits_{i=0}^{r-1} \frac{1}{W^i} \frac{1}{(r-1-i)!} \sum\limits_{l\in\mathbb{Z}} \sum\limits_{n=0}^{L-1} \Theta_{ni}(Wt - \rho l) \nonumber \int\limits_{t}^{\frac{x_n+ \rho l}{W}} g^{(r)}(u) \left(\frac{x_n+\rho l}{W} - u \right)^{r-1-i} du. \nonumber \\
\end{align*}
Since the interpolating kernels $\Theta_{ni}$ are compactly supported on $[-\rho+s+1, \mu+s]$ for all $n=0,\dots,L-1,~i=0,\dots,r-1$,  we have
\begin{multline}\label{5.5}
| (S_W g)(t) - g(t) | \leq \sum\limits_{i=0}^{r-1} \frac{1}{W^i(r-1-i)!} \\
\sum\limits_{n=0}^{L-1}
\sum_{|W t - \rho l| \leq T_s} \left| \int\limits_{t}^{\frac{x_n+\rho l}{W}} \left| g^{(r)}(u) \right| \left| \frac{x_n+\rho l}{W} - u \right|^{r-1-i} du \right|
\times \left|\Theta_{ni}(Wt - \rho l) \right|, 
\end{multline}
where  $T_s = \max \left\{|-\rho+s+1|,|\mu+s|\right\}.$ Since 
\begin{align*} 
\left| \displaystyle\int\limits\limits_{t}^{\frac{x_n+\rho l}{W}} |g^{(r)}(u)| \left| \frac{x_n+\rho l}{W} - u \right|^{r-1-i} du \right|
&\leq \left|\int_{t}^{\frac{x_n+\rho l}{W}} \left| g^{(r)}(u) \right| \left( \frac{T_s+s+1}{W} \right)^{r-1-i} du \right| \\
&=\left( \frac{T_s+s+1}{W} \right)^{r-1-i} \left|\int_{0}^{\frac{x_n+\rho l}{W} - t} \left| g^{(r)}(v + t) \right| dv \right|\\
&\leq\left( \frac{T_s+s+1}{W} \right)^{r-1-i} \int\limits_{|v| \leq \frac{T_s+s+1}{W}} \left| g^{(r)}(v + t) \right| dv,
\end{align*}
\eqref{5.5} becomes $|(S_W g)(t)- g(t)|$
\begin{align*}
 &\leq\sum\limits_{i=0}^{r-1}\frac{1}{W^i(r-1-i)!}  
\sum\limits_{n=0}^{L-1} \sum_{|W t - \rho l| \leq T_s} \left( \frac{T_s+s+1}{W} \right)^{r-1-i} \left|\Theta_{ni}(Wt - \rho l) \right|\int_{|v| \leq \frac{T_s+s+1}{W}} \left| g^{(r)}(v + t) \right| dv
\\
&\leq \sum\limits_{i=0}^{r-1} \sum\limits_{n=0}^{L-1} \frac{m_0(\Theta_{ni})}{W^i(r-1-i)!}  
\left( \frac{T_s+s+1}{W} \right)^{r-1-i} 
\int_{|v| \leq \frac{T_s+s+1}{W}} \left| g^{(r)}(v + t) \right| dv.
\end{align*}
Now by generalized Minkowski inequality, we have
\begin{align*}
\| S_W g - g \|_{p} &\leq \sum\limits_{i=0}^{r-1}\sum\limits_{n=0}^{L-1}\frac{m_0(\Theta_{ni})}{W^i(r-1-i)!} \left( \frac{T_s+s+1}{W} \right)^{r-1-i} \int_{|v| \le \frac{T_s+s+1}{W}} \| g^{(r)}(v+\bullet) \|_{p} \, dv \\
&= \sum\limits_{i=0}^{r-1}\sum\limits_{n=0}^{L-1}\frac{m_0(\Theta_{ni})}{W^i(r-1-i)!} \left( \frac{T_s+s+1}{W} \right)^{r-1-i} \frac{2(T_s+s+1)}{W} \| g^{(r)} \|_{p} \\
&= W^{-r}\sum\limits_{i=0}^{r-1}\sum\limits_{n=0}^{L-1}\frac{m_0(\Theta_{ni})}{(r-1-i)!} 2\left( T_s+s+1\right)^{r-i}  \| g^{(r)} \|_{p}\\
&= C W^{-r} \| g^{(r)} \|_{p},
\end{align*}
where $C>0$ is a constant independent of $W$. 

Thus, the sampling operator $S_W$ satisfies the hypotheses of our interpolation Theorem \ref{interpolation theorem}, which allows us to establish Theorem \ref{approximation operator}.\\\\

\section*{Acknowledgment}
The author (ST) gratefully acknowledges the Ministry of Education, Government of India, for supporting this research through the Prime Minister’s Research Fellowship and Grant (PMRF ID: 1603259).
\section*{Declarations}
\subsection*{Data Availability}
Data sharing does not apply to this article as no datasets were generated or analysed during the current study.
\subsection*{Conflict of interest} The authors declare that there is no conflict of interest.

\end{document}